%% file: SPTPRBFinal.tex
\documentclass[prb,reprint]{revtex4-1}
\usepackage{H1}
\usepackage{amsmath}
\usepackage{amssymb}
\usepackage{amsthm}
\usepackage{amsfonts}
\usepackage{bbm}
\usepackage{array}
\newcolumntype{L}[1]{>{\raggedright\let\newline\\\arraybackslash\hspace{0pt}}m{#1}}
\newcolumntype{C}[1]{>{\centering\let\newline\\\arraybackslash\hspace{0pt}}m{#1}}
\newcolumntype{R}[1]{>{\raggedleft\let\newline\\\arraybackslash\hspace{0pt}}m{#1}}

\newcommand{\mfa}{\mathfrak{a}}
\newcolumntype{N}{@{}m{0pt}@{}}

\newcommand{\zzt}{\mathbb{Z}_2}
\newcommand{\zzts}{\mathbb{Z}^{\otimes2}_2}

\begin{document}
\title{Phase Structure of 1d Interacting Floquet Systems I: Abelian SPTs }
\author{
C.~W.~von~Keyserlingk}
\address{
 Princeton Center for Theoretical Science, Princeton University, Princeton, New Jersey 08544, USA
}
\author{
S.~L.~Sondhi}
\address{
Department of Physics, Princeton University, Princeton, New Jersey 08544, USA
}
\begin{abstract}
Recent work suggests that a sharp definition of `phase of matter' can be given for some quantum systems out of equilibrium---first for
many-body localized systems with time independent Hamiltonians and more recently for periodically driven or Floquet localized systems. 
In this work we propose a classification of the finite abelian symmetry protected phases of interacting Floquet localized systems in one dimension. We find that the different Floquet phases correspond to elements of $\text{Cl}_G\times\mathcal{A}_G$, where $\text{Cl}_G$ is the undriven interacting classification, and $\mathcal{A}_G$ is a set of (twisted) 1d representations corresponding to symmetry group $G$. We will address symmetry broken phases in a subsequent paper.
\end{abstract}
\maketitle

\section{Introduction}
The past few years have seen considerable progress in our understanding of the phenomenon of many body localization (MBL) which 
has built on the early\cite{Basko06}, seminal and rigorous\cite{Imbrie14} work that established its existence. One of the more interesting ideas that has emerged
from this work is that of eigenstate phase transitions wherein individual many body eigenstates and/or the eigenspectrum exhibit 
singular changes in their properties across a parameter boundary even as the standard statistical mechanical averages are perfectly
smooth. In recent work \cite{Khemani15}  this idea was generalized to disordered Floquet systems taking advantage of the fact that they exhibit generalizations of the notions of eigenstate and eigenvalue in the form of time-periodic Floquet eigenstates and associated quasi-energies. Ref.~\onlinecite{Khemani15} presented evidence that one dimensional spin chains with Ising symmetry exhibit four
distinct Floquet phases with either paramagnetic or spin glass order -- two of the resulting Floquet phases have no analogs in undriven systems. Disorder seems to be an essential ingredient in this generalization  -- if the driving Hamiltonians are clean\cite{Lazarides14PRE,Rigol14,Abanin14,Ponte15},  or lack sufficiently strong disorder\cite{Lazarides14PRL}, the eigenstate properties of periodically driven systems seem to exhibit ``infinite temperature'' ergodic behavior, with no vestige of  paramagnetic or spin glass quantum order.


In this paper we pick up the thread from this point and address the question of obtaining an enumeration of all possible Floquet
phases in one dimension. Specifically, we restrict ourselves to Floquet phases which do {\it not} spontaneously break any symmetry of the
drive; we will analyze the case of broken symmetry in a subsequent paper. In 1d this implies that we are looking for Floquet
versions of symmetry protected topological (SPT) phases of matter, which generalize topological insulators and superconductors to
interacting systems.

All SPT phases of matter are associated with some global symmetry group $G$, which  is not spontaneously broken. Given a symmetry group $G$ there may be many distinct SPT phases, each of which can be distinguished by their ground states -- two ground states represent the same SPT phase iff they can be connected to one another by a symmetric local unitary. A complete
classification of SPTs in 1d is available \cite{Chen11,PollmannBerg11}. The first step away from the purely ground state classification was taken
in Refs.~\onlinecite{Chandran14,Bahri15} where it was shown that in the presence of localization induced by sufficiently strong disorder, the entire spectrum (not just the ground state) of certain SPTs can carry a signature of the underlying SPT order. In 1d this is the statement that in many-body localized SPT systems, the entire spectrum has a characteristic string order. This idea was clarified recently in Ref.~\onlinecite{Potter15}. On the one hand MBL Hamiltonians are believed to be characterized by the appearance of a complete set of local integrals of the motion or $l-$bits \cite{Serbyn13a,Serbyn13cons,Huse14}. On the other hand, it is known that a proper subset of the possible SPT orders can be captured by commuting stabilizer Hamiltonians\footnote{A commuting stabilizer Hamiltonian takes form $H=\sum_r H_r$, where the $H_r$ are local and commute amongst themselves. See Ref.~\onlinecite{Bahri15} for a well explained example of a commuting stabilizer SPT Hamiltonian, and examples in the main text.}. Ref.~\onlinecite{Potter15} synthesized these observations arguing that only those SPT orders captured by commuting stabilizer Hamiltonians can have eigenstate order.

In another line of work, non-interacting Floquet systems have been investigated for non-trivial topology and building on various
examples \cite{Kitagawa10,Jiang11,Lindner11,Thakurathi13,Rudner13,Asboth14,Carpentier15} a classification has recently been obtained \cite{Nathan15,Roy15}, and also investigated in a disordered setting\cite{Titum15a,Titum15b}. This classification is
indeed richer than for the undriven problem. As Ref.~\onlinecite{Nathan15} shows, if the original equilibrium non-interacting classification was $\text{Cl}=\mathbb{Z},\mathbb{Z}_2 \text{ or }\{0\}$\cite{Schnyder08}, then the Floquet classification is of form $\text{Cl}\times \text{Cl}$.

Here we will show that the classification of symmetric Floquet states is different from both the undriven MBL SPT classification {\it and} the non-interacting Floquet classification; a simple example is given by the $G=\mathbb{Z}_2 \times \mathbb{Z}_2$ bosonic SPT considered in detail in \secref{s:Z2xZ2}. Our general approach is as follows. We start with a Floquet drive with associated Floquet unitary $U_{f}\equiv U(T)$. We assume that $U_{f}$ has a prescribed eigenstate order which we call the bulk order (measured by a string order parameter for unitary $G$, and encoded in the form of the local conserved quantities). To focus on those Floquet systems potentially resilient against heating to ``infinite temperature'', we consider only those bulk orders which are many-body localizable in the sense of Ref.~\onlinecite{Potter15} -- in practice, this means we restrict ourselves to states with on-site symmetry groups $G$. We make a further simplification by assuming $G$ is abelian.  

In the undriven setting, the classification of 1d SPT eigenstate order is captured by how the symmetries in the problem act projectively at the edge\cite{Kitaev11,Chen11,PollmannBerg11}, this being captured almost entirely\footnote{For fermionic states, the fermion parity of the symmetry action at the edge is also important\cite{Kitaev11}.} by a so-called 2-cocycle $c(g,h)\in \text{U}(1)$\footnote{See also Ref.~\onlinecite{Pollmann10,PollmannBerg11} for a more pedestrian exposition, and  Ref.~\onlinecite{ChenScience} and \secref{s:algebraic} for an introduction to cocycles.}. We conjecture that  in the driven MBL setting in addition to this information there is just one further piece of data $\kappa:G\rightarrow  \text{U}(1)$ characterizing the commutation between the symmetry action local to the edges and the Floquet unitary $U_f$ itself. For unitary symmetries, we show that this $\kappa$ is a 1d representation of $G$. For anti-unitary symmetry groups of form $G=G'\times \mathbb{Z}^{\mathscr{T}}_2$ where $G'$ is unitary and $ \mathbb{Z}^{\mathscr{T}}_2$ is $\mathscr{T}^2=1$ time reversal, our results are less certain, but we conjecture that $\kappa$ obeys a twisted 1-cocycle condition \eqnref{eq:twistedrep}. In any case, the set of all such $\kappa$ is denoted $\mathcal{A}_G$. Hence our proposed interacting classification for Floquet drives is of the form $\text{Cl}\times\mathcal{A}_G$ where $\text{Cl}$ is the undriven classification. A compatible set of results was obtained independently shortly after the appearance of the present work in Refs.~\onlinecite{Else16,Potter16,Harper16}. 

 To support our conjecture we investigate in more detail the structure of local symmetric Floquet unitaries with a complete set of local bulk integrals of motion. On an open chain, we argue that such a unitary can be brought into a form $U_f=v_L v_R e^{- i f }$ where $v_L,v_R$ are local to the left/right end of the system and both commute with $f$, a local functional of the local bulk conserved quantities characterizing the bulk order. We formulate $\kappa$ in terms of $v_L,v_R$, and show that it is robust under arbitrarily large but symmetric modifications to the Floquet drive local to the edges as well as small symmetric perturbations to the bulk.  

The balance of the  paper is set out as follows. We start in \secref{s:motivatingexample} with  a 1d Floquet system with class D (fermion parity protected) Kitaev chain eigenstate order. We use a novel framework to reproduce the $\mathbb{Z}_2\times \mathbb{Z}_2$ classification obtained from band theory, and verified in an interacting setting in Ref.~\onlinecite{Khemani15}. This helps us to motivate a more general framework in \secref{s:generalities}. Then in \secref{s:algebraic} we reinterpret our results more simply as an extension of the undriven algebraic classification of Ref.~\onlinecite{Kitaev11}. In \secref{s:Z2xZ2} we examine the spectrum of the interacting bosonic $G=\zzts$ drive, explaining what kinds of edge modes are present. In \secref{s:BDI} we deal separately with two examples of time reversal invariant SPTs, the latter being an interacting driven version of fermionic Class BDI. We conclude in \secref{s:conclusion}. 

\section{Motivating example: Class D in 1d}\label{s:motivatingexample}
\begin{figure}[h]
\includegraphics[width=0.95\columnwidth]{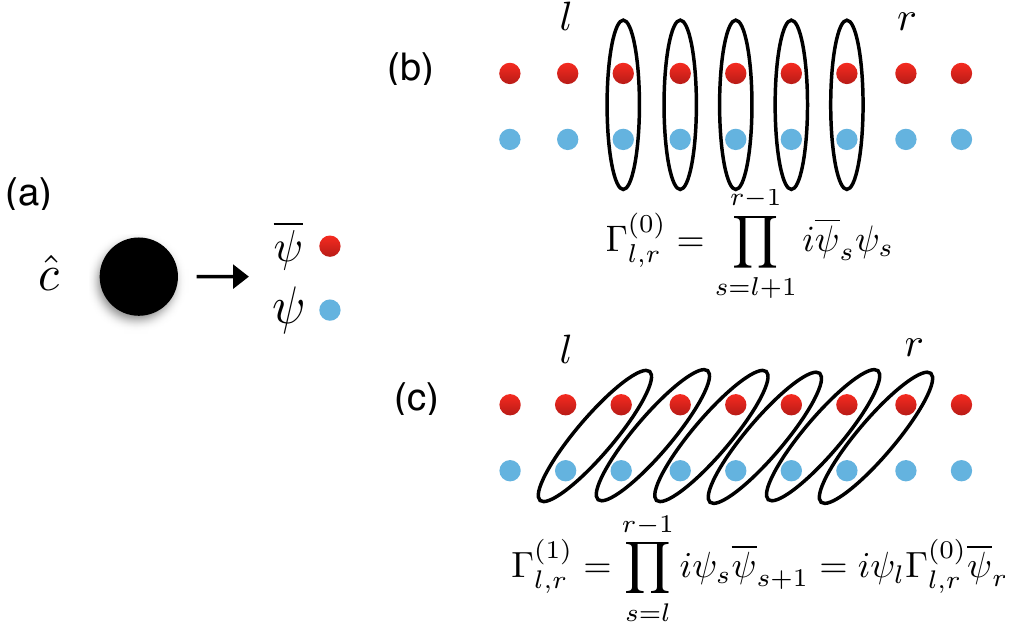}
\caption{(Color Online): (a) Each on-site spinless fermion is equivalent to two independent Majorana degrees of freedom. (b) The string order operator for the non-topological phase involves a product of onsite fermion parity operators $(-1)^{\widehat{c}^{\dagger}_s \widehat{c}_s}=i\psi_s \psib_s$, while the string order operator for the topological phase (c) involves a product of bond fermion parity operators $i\psi_s \psib_{s+1}$. }
\label{classDfigure}
\end{figure}

\begin{figure}[h]
\includegraphics[width=0.95\columnwidth]{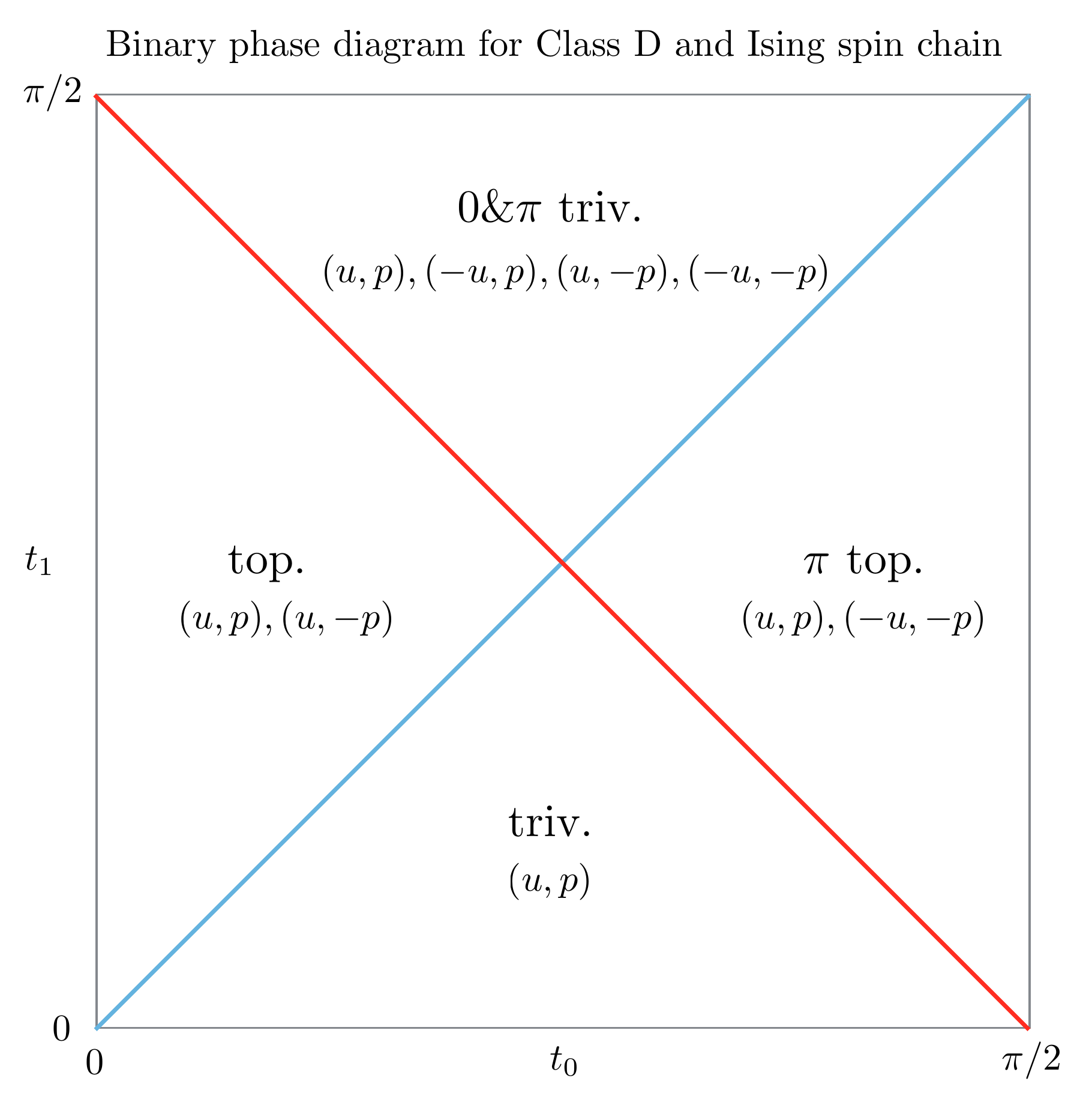}
\caption{(Color Online): This shows the phase diagram for the binary drive in \eqnref{eq:classDUf}. The red and blue line separate distinct Floquet phases. The lists involving $(u,p)$ summarize the protected multiplets in the spectrum for an open chain e.g., in the $0 \& \pi$ trivial phase, if there is a state with $U_f,P$ eigenvalues $(u,p)$  then there are guaranteed to be states at $(-u,p),(u,-p),(-u,-p)$ up to exponential corrections. } 
\label{classDPhasediagram}
\end{figure}

We will now consider in some detail a particular case that will explain the logic we follow in the general case. This is the case of Floquet
drives defined by fermionic Hamiltonians $H(t)$ which conserve fermion parity. For quadratic time independent Hamiltonians this is Class D
in the Altland-Zirnbauer classification. We will use the same nomenclature for our interacting time dependent problem. For the
non-interacting Floquet problem the list of phases for Class D is known and we will explain how this exhausts the list of phases in the
interacting MBL setting as well; we note that previous work \cite{Khemani15} has shown by explicit computation that the 
non-interacting phases do continue into this setting but not settled the question of whether others exist.   We now, successively,
review the basics of the time independent quadratic classification, its Floquet analog and an understanding of the latter appropriate
to the MBL setting and end with the promised generalization to interacting MBL Floquet systems.

\subsection{Time independent SPT phases}

There are just two SPTs with just $\mathbb{Z}^{\text{fp}}_{2}$ fermion parity
symmetry in 1d\cite{Chen11,Kitaev11}. They have model Hamiltonians

\begin{align*}
H_{0}= & -\sum^{N}_{i=1} ih_{i}\psi_{i}\bar{\psi}_{i}\\
H_{1}= & -\sum^{N-1}_{i=1} iJ_{i}\psi_{i}\bar{\psi}_{i+i}
\end{align*}
where the Hilbert space consists of $N$ spinless fermion degrees of freedom $\widehat{c}$, or equivalently two $\psi,\psib$ Majorana fermion degrees of freedom per site defined by $\widehat{c} = (\psi - i \psib )/2$. If we choose the $h_{i},J_{i}$ to be translationally invariant then $H_{0},H_{1}$ encode the well known Kitaev trivial/topological 1d wire fixed point states. Hamiltonians in the same phase as $H_1$ are called topological because they are associated with an (exponentially) protected spectral pairing on an open system associated with a protected
Majorana mode at its edge. Concretely,  for $H_{1}$ presented above, we can find simultaneous
eigenstates of $H_{1}$ and fermion parity $P\equiv \prod^N_{i=1} i \psib_i \psi_i$. The operator $\bar{\psi}_{1}$
commutes with $H_{1}$  but anti-commutes with $P$. Hence each energy
eigenvalue $E$ of $H$ is associated with at least two states $\mid E,+\rangle,\mid E,-\rangle=\bar{\psi}_{1}\mid E,+\rangle$ with fermion parity $\pm1$ respectively. $H_{0}$, the trivial state, has no such protected
degeneracies.

Another way to distinguish the ground states of the trivial/topological phases is through the use of string order parameters. That is, if we define
\be\label{eq:classDstringorder}
\Gamma_{l,r}^{(c)}=\begin{cases}
\prod_{l<s<r}i\psi_{s} \psib_{s} & c=0\\
i\psi_{l} \ds{\prod_{l<s<r}i\psi_{s} \psib_{s}} \psib_{r}& c=1
\end{cases}
 \ee
then $\langle\Gamma_{l,r}^{(0)}\rangle_{\text{gs}}$ is long-ranged/exponentially decaying in the trivial/topological phases respectively, and $\langle\Gamma_{l,r}^{(1)}\rangle_{\text{gs}}$ is long-ranged/exponentially decaying in the topological/trivial phases respectively (see Ref.~\onlinecite{Else13} for an illuminating account of string order in 1d SPTs). We will thus refer to the trivial/topological states as having type $0,1$ string orders respectively. If we choose $h_{i},J_{i}$ strongly disordered \textit{all of the eigenstates} of $H_{0},H_{1}$ necessarily have string order of types $0,1$ respectively, and this statement is at least perturbatively stable to the inclusion of interactions\cite{Potter15}. 

\subsection{Quadratic Floquet phases}

There are four known phases of the Class D Floquet problem. Following Ref.~\onlinecite{Khemani15} we can exhibit them in a simple
model of a binary Floquet drive (see also Ref.~\onlinecite{Thakurathi13}) using the reference Hamiltonians $H_{0},H_{1}$
\[
U(t)=\begin{cases}
e^{-iH_{0}t} & 0\leq t<t_{0}\\
e^{-iH_{1}(t-t_{0})}e^{-iH_{0}t_{0}} & t_{0}\leq t<t_{0}+t_{1}
\end{cases}
\]
where we pick $h_{i},J_{i}=1$ to be uniform. Eventually we will disorder these couplings, but we assume uniformity for now for ease of exposition. The final Floquet unitary is simply 
\begin{equation}\label{eq:classDUf}
U_f\equiv U(T)=e^{- i H_1 t_1 } e^{- i H_0 t_0 } \punc{.}
\end{equation}

The phase diagram of our binary drive as a function of $t0,t_1$ has some manifest periodicities. Note that $e^{\pi\psi_{i}\bar{\psi}_{i}}=-1$  whence the replacement $t_{0}\rightarrow t_{0}+n\pi$ shifts all quasi-energies by  $ N \pi$ but otherwise the Floquet eigenstate properties remain unchanged. The same holds true for shifts like $t_{1}\rightarrow t_1 +n \pi$.  Hence the eigenstate properties of $U_f$ are invariant under $t_i \rightarrow t_i+n_i \pi$. Another thing to note is that for systems with an even number of sites the unitary $i^{N/2}\prod_i \psi_i$ effectively flips $t_i \rightarrow -t_i$ while $\prod_{i \text{ even}} i\psi_i\bar{\psi}_i$ flips $t_1\rightarrow -t_1$. Therefore the eigenstate properties of $U_f$ are also invariant under such inversions and reflections in $t_0,t_1$. From this combination of shift and reflection symmetries in $t_{0,1}$, it suffices to consider a unit cell of the phase diagram  $t_{0,1} \in [0,\pi/2]$ as shown in \figref{classDPhasediagram}. The phase transition lines drawn in the diagram are straightforwardly obtained by diagonalizing $U(T)$
for closed chains where each individual momentum sector only presents a two dimensional problem. Of these, the boundary at small $t_{0},t_{1}$, can be obtained by using the BCH formula to show that $U_f\approx e^{-iH_{0}t_{0}-iH_{1}t_{1}}$. In this regime, the eigenstates are determined by the effective Hamiltonian $H_{0}t_{0}+H_{1}t_{1}$, which is expected to be fully trivial/topological
for $\left|t_{0}\right|>\left| t_{1}\right|$ and $\left|t_{0}\right|<\left|t_{1}\right|$ respectively.

We will now develop an analytical, spatially local, picture of this phase diagram, which in turn will guide our  attempt to classify 1d Floquet phases---we do this by focusing on the boundaries of the fundamental region where the Floquet unitaries will exhibit localization even
absent disorder. In the regions labelled `triv',  representative unitaries are obtained by setting $t_1=0$ i.e., $U_f=e^{-i H_0 t_0}$. It is clear that the eigenstate properties of these unitaries are simply those of the trivial Hamiltonian $H_0$ and {\it all} of the eigenstates have $c=0$ string order. This is clearly a consequence of the (trivial) localization of $H_0$. 
In the region labelled `top', representative unitaries are obtained by setting $t_0=0$ i.e., $U_f=e^{-i H_1 t_1}$, so the eigenstate properties of these are simply those of the topological Hamiltonian $H_1$.  All of the eigenstates have $c=1$ string order and on an open system, this drive will have a protected Majorana at its edges commuting with $U_f$, and a spectral pairing associated with this Majorana.

The $\pi$ topological phase is new to the driven setting. As an example, set $t_0 = \pi/2$ and $t_1=\epsilon < \frac{\pi}{2}$
$$
U_f = e^{-  i t_1 H_1} \prod_{s} i (i \psib_{s} \psi_{s}) \propto e^{-  i t_1 H_1} P \punc{.}
$$
This is simply the unitary associated with a topological drive (discussed above) multiplied by the global fermion parity operator which is itself a good quantum number of $H_1$. There is a complete basis of eigenstates with $c=1$ topological string order, which follows from the fact they have local integrals of the motion of form $i\psi_s \psib_{s+1}$. Now, rewrite $U_f$ in terms of said local integrals of the motion to obtain
$$
U_f \propto e^{-  i (t_1+\pi/2) H_1} i \psib_1 \psi_N \punc{.}
$$
This unitary looks like the topological drive above (with  $t_1$ shifted), multiplied by a term $i \psib_1\psi_N$ which commutes with the bulk local integrals of motion. Thus we can diagonalize the bulk unitary $e^{-  i (t_1+\pi/2) H_1}$ and the edge  unitary $U_{f,\text{edge}}=i \psib_1 \psi_N$ simultaneously, writing eigenvalues of $U_f$ as $u = u_b u_e$ where $u_b \in \text{U}(1),u_e=\pm 1$ are the eigenvalues of  the bulk and edge unitaries respectively. Note that $\psib_1$ anti-commutes with $U_{f,\text{edge}}$ and commutes with the bulk unitary. Hence  if $u=u_b u_e$ is an eigenvalue of $U_f$ then so is $-u = u_b \times(-u_e)$. In other words if $U_f \mid  u \rangle=u  \mid  u \rangle$ then $ \psib_1\mid  u \rangle$ has eigenvalue $-u$. This $\pi$ shift symmetry in the argument of $u$ associated with the boundary Majoranas is the reason we call this state the $\pi$ topological state. 

In summary, we identified the eigenstate order of the drive, wrote the unitary in terms of the corresponding local integrals of motion. The resulting unitary looked like a simple topological drive $e^{-  i (t_1+\pi/2) H_1}$ multiplied by a term $i \psib_1 \psi_N$ which hops Majoranas between the distant edges. This implied a spectral pairing at quasi-energy $\pi$.
 
We can treat the $0\&\pi$ phase analogously. On the boundary of that region, the eigenstates have type $0$ (trivial) string order \eqnref{eq:classDstringorder} and at a given edge, there is a $0$ and a $\pi$ quasi-energy Majorana mode. To see this, set $t_0=\epsilon < \frac{\pi}{2}$ and $t_1=\pi/2$. The resulting unitary simplifies to  
\begin{align}
U_f &=\prod^{N-1}_{s=1} i (i \psi_{s} \psib_{s+1}) e^{-  i t_0 H_0} \nonumber \\
& \propto P \psib_1e^{-  i t_0 i\psib_1  \psi_1}  \psi_N  e^{-  i t_0 i\psib_N  \psi_N} \times e^{-  i t_0 \sum^{N-1}_{s=2} i\psib_s  \psi_s} \nonumber\\
&  \rightarrow  P \psib_1 \psi_N  \times e^{-  i t_0 \sum^{N-1}_{s=2} i\psib_s  \psi_s}\label{eq:classDlocun} \punc{.}
\end{align}
In the last line we used a local symmetric unitary change of basis (implemented by $W = e^{ \frac{ i t_0}{2} \db{ i\psib_1  \psi_1+ i\psib_N  \psi_N}}$). Note that the on-site fermion parities $i\psib_s \psi_s $ are local integrals of motion in the bulk ($s=2,\ldots,N-1$). As in the previous example, we use these local integrals of motion to re-express the unitary as
\be
U_f \propto   i \psib_1   \psi_N  e^{-  i (t_0+\pi/2) \sum^{N-1}_{s=2} i\psib_s  \psi_s} \punc{.}
\ee
This looks like a bulk non-topological drive multiplied by a Majorana tunneling operator $i \psib_1   \psi_N$. Note that the edge degrees of freedom are completely decoupled from the bulk so we can simultaneously diagonalize the bulk   $e^{-  i (t_0+\pi/2) \sum^{N-1}_{s=2} i\psib_s  \psi_s}$ unitary and the two site edge unitary 
 $$
 U_{f,\text{edge}} =  i\psib_1   \psi_N =  i e^{-i\frac{\pi}{2} i \psib_1   \psi_N }\punc{.}
 $$
Note that the two boundary sites $1,N$  involve the \textit{four} Majoranas  $\psib_1,\psi_1,\psib_N,\psi_N$. This two-site unitary has two useful independent integrals of motion $ U_{f,\text{edge}} = i \psib_1 \psi_N$  and $ P_{\text{edge}} = i \psib_1  \psi_1 i\psib_N  \psi_N$ -- note these are also integrals of motion of the original unitary $U_f$ as well. Picking a reference eigenstate $\mid 1,1\rangle$ of $U_{f,\text{edge}}, P_{\text{edge}} $ for the two site problem, we can toggle between the four eigenstates of $U_{f,\text{edge}}$ as shown in \tabref{tab:classDreps}. 

\begin{table}[t]
\caption{Eigenstates of the unitary $U_{f,\text{edge}}$ involving sites $1,N$.}
\begin{center}
\begin{tabular}{|c|c|c|}
\hline 
 & $U_{f,\text{edge}}$ & $P_{\text{edge}}$\tabularnewline
\hline 
$\mid1,1\rangle $ & $1$ & $1$\tabularnewline
\hline 
$\psi_1 \mid1,1\rangle $ & $1$ & $-1$\tabularnewline
\hline 
$\psib_1 \mid1,1\rangle  $ & $-1$ & $-1$\tabularnewline
\hline 
$\psib_1 \psi_1 \mid1,1\rangle  $ & $-1$ & $1$\tabularnewline
\hline 
\end{tabular}
\end{center}
\label{tab:classDreps}
\end{table}%

Note that the edge degrees of freedom have two eigenstates at each of $U_{f,\text{edge}}=\pm 1$. It is straightforward to use the edge properties listed in \tabref{tab:classDreps} to show that the eigenstates of the full Floquet unitary come in quadruplets with $U_f,P$ eigenvalues $(u,p),(u,-p),(-u,-p),(-u,p)$. From  \tabref{tab:classDreps} we see that $\psi_1$ is associated with a flip $(u,p)\rightarrow (u,-p)$ while $\psib_1$ is associated with $(u,p)\rightarrow (-u,-p)$. Hence we think of $\psi_1,\psib_1$ as being a zero/$\pi$ quasi-energy Majoranas respectively.

Finally we can offer some intuition regarding the somewhat physically opaque constructions. The non-trivial drives (i.e., the $\pi$ and $0\&\pi$ drives) are associated with a tunneling operator of the form  $\psi_1 \psib_N$.  We can think of these operators as pumping fermion parity charge from one edge to the other, across the entire system. So non-trivial Floquet drives differs from trivial Floquet drives insofar as a charge of the symmetry group $G=\mathbb{Z}_2$ has been pumped across the system.

\subsection{Generalizing to the MBL regime}\label{ss:towards}
We were able to understand the specific class D drives above by reducing the Floquet unitary to the form
\be
U_f=v_L v_R e^{- i f } \label{eq:decomprough}
\ee
where $f$ is some functional of local bulk l-bits $\Gamma_s $, and $v_L,v_R$ are operators localized at the left and right edges of the system respectively which commute with all bulk $\Gamma_s$. For the non-trivial Floquet drives (i.e., the $\pi$ and $0\&\pi$ drives),  $v_L,v_R$ were both fermion parity odd unitary operators.

This begs the question: Can we always reduce $U_f$ to this simple form, and does the fermion parity of $v_L,v_R$ always indicate whether or not the Floquet drive is trivial? We claim yes.  In full,  we will argue that fermion parity symmetric Floquet unitaries with a complete set of bulk local integrals of motion and an associated trivial/topological eigenstate order : {\bf(i)} can be written as \eqnref{eq:decomprough}; {\bf(ii)} $v_L,v_R$ have definite (and identical) fermion parity as operators, and {\bf(iii)}; the parity of $v_L,v_R$ is uniquely determined by $U_f$ and robust to arbitrarily large parity symmetric modifications to the unitary near the edges, and sufficiently small bulk perturbations.

For {\bf(i)} it is essential to assume the existence of local conserved quantities -- this is presumably necessary at the outset if we wish for our system to not heat up in the sense of Refs.~\onlinecite{Lazarides14PRE,Rigol14,Abanin14,Ponte15}. In \appref{app:flocal} we argue that  $f$ can be chosen to be a local function of these local conserved quantities because $U_f$ arises from a time dependent local Hamiltonian. Before diving into the fuller discussion of the more general case in \secref{s:generalities}, and assuming {\bf(i)}, let us give some flavor of the arguments for {\bf(ii)},{\bf(iii)}.  Using \eqnref{eq:decomprough}, and the fact that $P$ commutes with both $U_f$ and the bulk l-bits, it follows that $[P:v_L v_R]=1$ where we define 
$$
[A:B]\equiv ABA^{-1} B^{-1}\punc{.}
$$
for invertible operators $A,B$. Now, as $P$ is a local unitary circuit (of depth 1), it is straightforward to see that $[P:v_{L/R}] $ is a unitary localized near the left/right end of the system respectively. Indeed we can write $[P:v_{L}] = v_{L} \theta_{L}$ and $[P:v_{R}] = \theta_{R} v_{R} $ where $\theta_{L/R}$ is some unitary operator localized near the left/right end of the system. $[P:v_L v_R]=1$ implies 
\be\label{eq:FP}
v_L v_R = v_L \theta_{L} \theta_{R}  v_R  \implies  \theta_{L} \theta_{R}=I \punc{.}
\ee
But $\theta_{L,R}$ are local to $L,R$ respectively, distant from one another, and yet inverse to one another. The only possibility is that $\theta_{L} = \theta_{R}^{-1}=e^{i \theta} I$ is a pure phase. Using the fact that $P^2=1$ it moreover follows that $e^{i \theta} = \pm 1$. Thus, $v_L,v_R$ have definite and equal fermion parities. 

For {\bf(iii)} we need to show that the parities remain unchanged if we augment our drive $U_f \rightarrow U_f w_{L,R} $ where $w_{L,R}$ are parity symmetric unitaries localized near the left/right edges of the system respectively. Heuristically speaking, we are just modifying $v_{L,R}\rightarrow v_{L,R} w_{L,R}$, which will not change the fermion parities of $v_{L,R}$ because $w_{L,R}$ are parity symmetric -- of course, this is a little misleading, because the modified $v_{L,R}$ do not necessarily commute with all of the bulk l-bits as required in \eqnref{eq:decomprough}. A fuller argument is provided in \secref{ss:Robust}. See also \appref{app:pumpedcharge} and \appref{s:classDedge} for a distinct and potentially tighter argument using string order parameters. 

Last, we wish to argue that the parity of $v_{L,R}$ is robust to sufficiently small bulk perturbations. This statement is supported by the observation that in the non-interacting Floquet setting with a random disorder configuration, the $0$ and $\pi$ Majoranas eventually decay into the bulk with probability $1$. The decay length is determined by the average behavior of the random couplings (see Ref.~\onlinecite{Gannot15} for an example of such a calculation, using transfer matrices). Upon modifying the bulk couplings smoothly, the decay length changes smoothly, and for sufficiently small changes the Majorana edge mode is robust with probability $1$. Thus in the non-interacting setting, the edge structure is at least statistically robust to small adjustments to the bulk. In our formalism when the $v_L,v_R$ are parity odd, they are the many body analogues of the $\pi$ Majoranas in the non-interacting setting. In this case we expect a similar statistical statement to hold. Namely, upon modifying the bulk couplings slightly, the $v_L,v_R$ operators remain fermion parity odd, and localized to the edges with probability $1$.

\section{General Framework}\label{s:generalities}
The discussion in the previous sections focussed on a fermion parity symmetric system $G=\mathbb{Z}^{\text{fp}}_2$. Here we consider more general Floquet drives with an on-site finite global abelian symmetry group $G$ with global generators $V(g)$. Start with some spatially local $G$-symmetric and time-periodic family of Hamiltonians $H(t) = H(t+T)$, giving rise to an instantaneous unitary $U(t)=\mathcal{T} e^{-\int^{t}_{0} dt' H(t')}$. Our aim is to characterize the eigenstates of  $U_{f}\equiv U(T)$  on a system with edges.

We assume that $U_f$ has full eigenstate order i.e., assume that the disorder in the drive is sufficiently strong such that there is a complete set of local integrals of the motion (l-bits) $\{\Gamma_r\}$, which encode a known unique SPT order. As discussed in the introduction, only certain SPT eigenstate orders are expected to exist stably as the eigenstate orders of Floquet unitaries. For this reason, we restrict our attention to such `many-body localizable' SPT orders from the outset. In 1d Ref.~\onlinecite{Potter15} suggest all SPTs with finite on-site symmetry group are many-body localizable. It is for this reason we consider finite discrete $G$, and for simplicity we focus on abelian $G$.

We make a technical assumption about the l-bits: we assume they can be chosen to commute with the global symmetry generators. This latter requirement is certainly true for 1d fixed point MBL abelian SPT phases\cite{Potter15}. Perturbing symmetrically away from the fixed point models, we expect the l-bits to be smeared out $\Gamma_r \rightarrow  \Gamma'_r$, in such a way that $\Gamma'_r $ also commutes with the global symmetries. While the l-bits at the fixed points are exactly local, the $\Gamma'_r$ are only exponentially local\cite{Lieb72}.

The more general arguments in this section will follow the same format as before in the class D case. We argue that finite abelian $G$ symmetric Floquet unitaries with a complete set of bulk l-bits and an associated trivial/topological eigenstate order : {\bf(i)} Can be written as \eqnref{eq:decomprough} (with $v_{L,R},f$ obeying the conditions stated below \eqnref{eq:decomprough}); {\bf(ii)} $v_L,v_R$ can be associated with a certain (twisted) 1d representation $\kappa$ of $G$ to be defined; {\bf(iii)} this 1d representation is uniquely determined by $U_f$ and robust to arbitrarily large $G$ symmetric modifications to the unitary near the edges, and sufficiently small bulk perturbations.  {\bf(i)-(iii)}  together suggest that the symmetry protected features of Floquet drives with the mentioned properties are captured by the bulk order and the (twisted) 1d representation $\kappa$. Therefore, labelling the possible bulk orders by Cl$_G$ and the possible (twisted) 1d representations by $\mathcal{A}_G$, we conjecture that the interacting Floquet classification is Cl$_G\times\mathcal{A}_G$  for the $G$ considered here.

This section is organized as follows. We supply the arguments for {\bf(i)} in \appref{app:misclemmas}. In the next two sections we show {\bf(ii)}, i.e., how to associate $v_{L,R}$ with a certain 1d (twisted) representation of $G$. This will involve demonstrating that the quantity
\be\label{eq:pcharge}
\kappa_{L,R}(g)\equiv(V(g):v_{L,R})
\ee
which we call the `pumped charge', defines a (twisted) 1d representation of $G$. Our discussion is split between the unitary and anti-unitary cases in \secref{ss:unitary} and \secref{ss:antiunitary} respectively. \eqnref{eq:pcharge} involves a generalization of the group commutator defined by
\be \label{eq:groupcomm}
(V(g):W)\equiv V(g) W  V^{-1}(g) W^{-\alpha(g)}
\ee
where $W$ is any unitary, and $\alpha: G \rightarrow \mathbb{Z}_2$ is a homomorphism with $\alpha(g)=\pm 1$ for $g$ unitary/anti-unitary respectively. Note that for unitary $g$, $(V(g):W)=[V(g):W]$. 

In \secref{ss:Robust} we argue {\bf(iii)}, i.e., $\kappa$ is well defined and robust to arbitrarily large symmetric adjustments to the unitary local to $L,R$, and sufficiently small bulk perturbations. Finally in \secref{ss:summary} we summarize our proposed  classification, and provide some examples.

\subsection{Unitary on-site symmetry groups}\label{ss:unitary}Recall that $U_f$ is generated by a symmetric family of Hamiltonians $H(t')$. This implies that $U_f$ commutes with the global symmetry generators i.e., $(V(g):U_f)=1$ for all $g\in G$. We use this fact to constrain the symmetry properties of the $v_L,v_R$ appearing in \eqnref{eq:decomprough}. 

\begin{lemma}\label{lem:phase}
Consider system $S=[L,R]$, and finite abelian unitary symmetry group $G$. Then $\kappa_{L,R} (g)\equiv (V(g):v_{L,R})$ are $\text{U}(1)$ scalar operators, where $V(g)$ is a global symmetry transformation.
\end{lemma}
\begin{proof}
From assumption {\bf(i)} we argued that the Floquet unitary takes form \eqnref{eq:decomprough} where $v_L,v_R$ are localized at the $L,R$ edges, and commute with $f$. The unitary is $G$-symmetric so $(V(g):U_f)=1$ for all $g\in G$.  In addition, $f$  commutes with the global symmetry generators because it is a function only of l-bits, which are all assumed to commute with the global symmetries. Therefore
\be\label{eq:comm1}
1= (V(g):v_L v_R e^{-i f} )=(V(g):v_L v_R)\punc{.}
\ee 
The RHS of this equation can be expressed as (abbreviating $V(g)$ to $V$ and $\kappa_{L,R}(g)$ to $\kappa_{L,R}$)
\begin{align}
(V:v_L v_R)&= \kappa_L v_L \kappa_R   v_R (v_L v_R)^{-1}\nonumber\\
&=\kappa_L  \kappa_R v_L v_R (v_L v_R)^{-1}\nonumber\\
&=\kappa_{L}\kappa_{R}\punc{.}\label{eq:kLkRone}
\end{align}
The second equality follows  because $v_L$ and $\kappa_{R}$ commute. To show this, it suffices to note {\bf(a)} both terms are localized at the $L,R$ edges respectively, and  {\bf(b)}  in a fermionic system at least one of these terms is fermion parity even. {\bf(a)} follows from the fact $V(g)$ is a low-depth unitary, so $\kappa_{L,R} (g)\equiv (V(g):v_{L,R})$ are local to the $L,R$ part of the system respectively.  {\bf(b)} follows from the earlier argument around \eqnref{eq:FP} that $v_{L,R}$ have definite parity $p=\pm 1$, so that $\kappa_{R}$ has parity $p^2=1$ as required. \eqnref{eq:comm1} and \eqnref{eq:kLkRone} together imply that $\kappa_L \kappa_R = I$. But then $\kappa_{L,R}$ are unitaries with support far away from one another, yet have $\kappa_L = \kappa^{-1}_R$. The only possibility is that $\kappa_{L,R}$ are of the forms $e^{\pm i \theta}1$ respectively.
\end{proof}
We define the pumped charge of $U_f$ to be the quantity $\kappa\equiv\kappa_L: G \rightarrow\text{U}(1)$ in \eqnref{eq:pcharge}. Using the fact that $\kappa$ is a phase from \lemref{lem:phase}, it follows from the definition that
\be
\kappa(gh) = \kappa(g) \kappa(h)
\ee
with $\kappa(1)=1$, so that $\kappa$ forms a 1d representation for unitary symmetry groups. 
\subsection{Time reversal ($\mathscr{T}$) symmetry}\label{ss:antiunitary}
We can apply most of the arguments in the above subsection to drives with symmetry group $G=G'\times \mathbb{Z}^{\mathscr{T}}_2$ where $G'$ is unitary and $ \mathbb{Z}^{\mathscr{T}}_2$ is $\mathscr{T}^2=1$ time reversal, but there are a few complications. First in \lemref{lem:phase} we used the fact that $(V(g):U_f)=1$. This follows for unitary symmetries because each of the instantaneous Hamiltonians $H(t')$ are symmetric from which it readily follows that $U_f$ is symmetric. In contrast, if $H(t')$ is time reversal symmetric then $U_f$ will not necessarily obey $(V(\mathscr{T}):U_f)=1$.  However, if we insist in addition that $H(t')= H(T-t')$\cite{Nathan15} where $T$ is the period of the drive, then $(V(g):U_f)=1$ is guaranteed for all $g\in G$ including $g=\mathscr{T}$.

The next stumbling point in attempting to formulate an analogue of \lemref{lem:phase} is in showing that 
\begin{equation}\label{eq:key2}
(V(g):e^{-i f})=1,
\end{equation}
where $U_{\text{bulk}} = e^{-i f}$ is the bulk part of the unitary \eqnref{eq:decomprough}. While \eqnref{eq:key2} is clear in the unitary case it is less clear in the non-unitary case, and in fact we will only prove it for a subset of the time reversal invariant SPTs. To see where the problem arises, express 
\begin{equation}
 U_{\text{bulk}} =  \sum_{ J } \delta(\{\Gamma\}= J )  \underbrace{e^{-i f( J )}}_{\beta( J )} \nonumber
\end{equation}
where $\sum_{ J }$ is a sum over all the possible values for all l-bits, and $\beta( J ) $ are necessarily U$(1)$ numbers. By assumption the global time reversal generator $V(\mathscr{T})$ commutes with all the l-bits $\Gamma$. Then \eqnref{eq:key2} holds provided $\beta(J)= \beta(J^*)$. One way to guarantee this is to consider only those SPT orders for which the eigenvalues are real e.g., $J=\pm 1$.  We will assume this condition, as it is automatically true in the cases we wish to consider in this text \secref{s:BDI}. 

Assuming then that \eqnref{eq:key2} is true, \eqnref{eq:comm1} must hold. We now revisit the argument in \lemref{lem:phase}  to find (abbreviating $\alpha(g)$ to $\alpha$, $V(g)$ to $V$, and $\kappa_{L,R}(g)$ to $\kappa_{L,R}$)
\begin{align}
(V:v_L v_R)&= \kappa_L v^{\alpha}_L \kappa_R  v^\alpha_R (v_L v_R)^{-\alpha}\nonumber\\
&=(V:v_L )  (V:v_R ) v^{\alpha}_L v^\alpha_R (v_L v_R)^{-\alpha}\nonumber\\
&=(V:v_L )  (V:v_R ) p^{\frac{1-\alpha}{2}}\punc{.}
\end{align}
 The second equality follows from the fact that $\kappa_R$ is fermion even and localized to the right hand edge as before. The third equality follows from the fact $v_{L,R}$ have definite and identical fermion parity $p$ as before. Hence, using \eqnref{eq:comm1} we find that $\kappa_L \kappa_R = p^{\frac{1-\alpha}{2}} I$. But, as $\kappa_{L,R}$ have support far away from one another we must have
\be \label{eq:kappaconstr2}
\kappa(g) \equiv \kappa_L(g) = \kappa^{-1}_R(g) p^{\frac{1-\alpha(g)}{2}}
\ee
is a pure phase as before. It again follows readily that $\kappa(1)=1$. One slight difference however is that 
\begin{align}
&\kappa(gh)\nonumber\\
&=  V(g) \kappa(h) v^{\alpha(h)}_L V(g)^{-1}  v_L^{-\alpha(gh)}\nonumber\\
&= \kappa^{\alpha(h)}(g) \kappa^{\alpha(g)}(h)\punc{.}\nonumber
\end{align}
Therefore the analogue of pumped charge in this non-unitary case obeys 
\be\label{eq:twistedrep}
\kappa(gh) = \kappa(g)^{\alpha(h)} \kappa(h)^{\alpha(g)}\punc{,}
\ee
so that $\kappa$ is a `twisted' analogue of a 1d representation of the group $G$.
\subsection{Robustness of pumped charge}\label{ss:Robust}
Having associated $v_L,v_R$ with (twisted) 1d representation $\kappa$, we now show {\bf(iii)}, namely that $\kappa$ is: {\bf(a)} well defined i.e., independent of the precise manner in which we decompose \eqnref{eq:decomprough}; {\bf(b)} robust to symmetric modifications of the unitary at the edges; and {\bf(c)} robust to sufficiently small symmetric bulk perturbations.  

To show {\bf (a)}, suppose we have two decompositions $U_f= v_L v_R e^{- i f} = v'_L v'_R e^{- i f'}$ obeying the conditions below \eqnref{eq:decomprough}. Are the pumped charges the same? We argue yes. First we restrict attention to those terms in $f$ involving only conserved quantities in some extensive connected sub-region of the bulk $S'$ which is nevertheless far away from both $L,R$, forming functional $f_{S'}$.  From the locality of $f$ it follows that  $e^{i f_{S'}} e^{- i f}$ acts like the identity over almost all of $S'$. Indeed 
\be\label{eq:flfr}
e^{i f_{S'}} e^{- i f} = e^{- i f_L}  e^{- i f_R}\punc{,}
\ee
where $f_L,f_R$ are functions of bulk l-bits on left/right parts of the system respectively which are widely separated from one another. Now $v_L$ commutes with $f$ by assumption, and commutes with $f_{S'}$ because it is supported far away from $S'$. It follows readily from \eqnref{eq:flfr} that $v_{L}$ commutes with $e^{-i f_{L}}$, a fact we will use shortly. Now, as $e^{- i f},e^{-i f'}$ act identically deep in the bulk it is also true that $e^{i f_{S'}} e^{- i f'}$ acts like the identity over most of $S'$, and it has a similar decomposition $e^{- i f'_L}  e^{- i f'_R}$ into unitaries based at the left and right sides of the system, and as before $v'_L$ commutes with $e^{- i f'_L}$.  Now multiplying $U_f$ by $e^{i f_{S'}}$ we obtain
$$
e^{i f_{S'}} U_f= v_L e^{- i f_L} v_R  e^{- i f_R} = v'_L e^{-i f'_L} v'_R  e^{-i f'_R}\punc{.}
$$
As $v_{L} e^{- i f_{L}},v'_{L} e^{- i f'_{L}}$ have support far away from $v_{R} e^{- i f_{R}},v'_{R} e^{- i f'_{R}}$, it follows that 
$$
v_{L} e^{-i f_{L}} =  v'_{L} e^{-i f'_{L}}  e^{i \phi}
$$
where $e^{i \phi}$ is some U$(1)$ phase, using the same reasoning deployed below \eqnref{eq:FP} and in \lemref{lem:phase}. Applying $(V(g):\cdot)$ to this equation shows $v_L,v'_L$ have the same pumped charge because bulk l-bits commute with the global symmetry generators, and $v_L,v'_L$ commute with $e^{-i f_{L}},e^{-i f'_{L}}$ respectively. Thus we have argued that the pumped charge is well defined. 

For {\bf(b)} we need to show that the pumped charge is robust to local symmetric changes at the edge. Under such modifications, the resulting unitary will still have a complete set of conserved quantities deep the bulk so $U_{f}$ still has bulk eigenstate order. We just need to show that the pumped charge is robust under modifications of form $U_{f}\rightarrow U_{f} w_{L}$ for some symmetric $w_{L}$ localized near (say) the left edge such that $(V(g):U_f)=1$ continues to hold. Under such a modification $v_L$ will change, as will some of the l-bits near the left end of the chain. However for large system size, $v_R$ should not change under such a modification, and so neither does $\kappa_R$. Hence from the constraints \eqnref{eq:kappaconstr2} between $\kappa_{L,R}$, the pumped charge cannot change. For an alternative and perhaps more rigorous characterization of the pumped charge for systems with unitary symmetry group $G$, which does not require the knowledge that we can decompose $U_S = v_L v_R e^{-i f(\{\Gamma\})}$, see \appref{app:pumpedcharge}. Unfortunately,  in the anti-unitary case, some of the methods of \appref{app:pumpedcharge} are inapplicable because of the well known\cite{Chen15} difficulties in defining a `local time reversal' string operator.
%

Last we come to the slippery issue of whether the pumped charge is robust to sufficiently small bulk perturbations, and we give a similar argument as for the class D subclass in \secref{ss:towards}. Our expectation is that for a random disorder configuration, $v_{L,R}$ are with probability $1$ localized to the edges with a localization length determined in part by the spatially averaged values of the local couplings comprising $H(t)$. Under sufficiently small changes to these local couplings, we therefore expect the localization length and $v_{L,R}$ to change smoothly. For truly small changes in $v_L$ and $v_R$, the pumped charge $(V(g):v_L)$ being discrete (a  twisted 1d representation) cannot change and remains fixed.

\subsection{Summary and examples}\label{ss:summary}
Thus, for 1d Floquet drives with finite on-site abelian symmetry group $G$ and paramagnetic bulk order, our proposed Floquet classification looks like $\text{Cl}_G \times \mathcal{A}_G$ where $\text{Cl}_G$ is the undriven paramagnetic classification, and $\mathcal{A}_G$ consists of all of the 1d (twisted) representations of group $G$. \tabref{tab:class} gives examples.

The undriven classification $\text{Cl}_G$ can be read off from existing results \cite{Kitaev11,Chen11}. For a unitary abelian $G$, the 1D representations are in bijective correspondence with $G$ itself. So the Floquet classification takes the form $\text{Cl}_G \times G$ for finite unitary abelian groups. For class D, we see that $\mathcal{A}_{\mathbb{Z}^{\text{fp}}_2} = \mathbb{Z}_2$, so our scheme (which applies to interacting systems) reproduces the $\mathbb{Z}_2\times \mathbb{Z}_2$ Floquet classification result seen in the non-interacting band theory context\cite{Nathan15}. There are numerous examples of groups where $\mathcal{A}_G  \neq \text{Cl}_G $, so that our prediction breaks the $ \text{Cl}_G  \times  \text{Cl}_G $ pattern seen in the non-interacting classification up until now\cite{Nathan15}. For example, for  $G=\mathbb{Z}_2\times \mathbb{Z}_2$ one obtains $\text{Cl}_G=\mathbb{Z}_2$ and $\mathcal{A}_G=\mathbb{Z}_2\times\mathbb{Z}_2$ -- see \secref{s:Z2xZ2} for a description of the model, and examples of drives.

On the other hand, for symmetry groups with time reversal of the form $G=G'\times\mathbb{Z}^{\mathscr{T}}_2$ where $G'$ is finite on-site unitary, we show in \appref{app:twisted1dreps} that the possible twisted 1d representations $\kappa$ are precisely $\mathcal{A}_G=H^1(G', \text{U}(1))\times \mathbb{Z}_2$  i.e., specified by a $1$d unitary representation of $G'$ and a choice of $\pm 1$. For abelian $G'$ this implies a Floquet classification of form $\text{Cl}_G\times G' \times \mathbb{Z}_2$. For example, for $G= \mathbb{Z}^{\mathscr{T}}_2$ we get Floquet classification $\mathbb{Z}_2 \times \mathbb{Z}_2$. On the other hand for a `BDI' fermion system with $G= \mathbb{Z}^{\mathscr{T}}_2 \times \mathbb{Z}^{\text{fp}}_2$ we obtain Floquet classification $\mathbb{Z}_8\times \mathbb{Z}_2\times\mathbb{Z}_2$ -- although, see  \secref{sss:stacking} where we argue that in certain regards the classification can be regarded as $\mathbb{Z}_8\times \mathbb{Z}_4$. The free fermion Floquet classification of the BDI system gives $\mathbb{Z}\times \mathbb{Z}$, so we observe an interaction induced breaking of results similar to that seen in the non-driven context Ref.~\onlinecite{Kitaev11}. We direct the reader to \secref{s:BDI} for a further discussion of this point, and examples of the different drives. As an aside we note that  Note that, while our formalism includes the possibility of anti-unitary symmetries, the notion of eigenstate order and MBL in these scenarios has not yet been established in detail (see discussion in Ref.~\onlinecite{Potter15}). For this reason our results with non-unitary groups should be considered with caution. 
 
 \begin{table}
\begin{tabular}{|C{1.6cm}|C{1.5cm}|C{1.5cm}|C{1.7cm}|N}
\hline 
On site Symm. (G) & Undriven classif. (Cl$_G$)  & Twisted 1d reps. ($\mathcal{A}_G$) & Floq. MBL PM classif. ($\text{Cl}_G\times\mathcal{A}_G$) \tabularnewline
\hline 
$\mathbb{Z}_{2}^{\text{fp}}$ & $\mathbb{Z}_{2}$ & $\mathbb{Z}_{2}$ & $\mathbb{Z}_{2}\times\mathbb{Z}_{2}$\tabularnewline
\hline 
$\mathbb{Z}_{2}^{\text{}}\times\mathbb{Z}_{2}$ & $\mathbb{Z}_{2}$ & $\mathbb{Z}_{2}\times\mathbb{Z}_{2}$ & $\mathbb{Z}_{2}\times\mathbb{Z}_{2}\times\mathbb{Z}_{2}$\tabularnewline
\hline 
$ \mathbb{Z}_{2}^{\mathscr{T}}$ & $\mathbb{Z}_{2}$ & $\mathbb{Z}_{2}$ & $\mathbb{Z}_{2}\times\mathbb{Z}_{2}$\tabularnewline
\hline 
$\mathbb{Z}_{2}^{\text{fp}}\times \mathbb{Z}_{2}^{\mathscr{T}}$ & $\mathbb{Z}_{8}$ & $\mathbb{Z}_{2}\times\mathbb{Z}_{2}$ & $\mathbb{Z}_{8}\times\mathbb{Z}_{2}\times\mathbb{Z}_{2}$\tabularnewline
\hline 
\end{tabular}
\caption{This table gives examples of our proposed $\text{Cl}_G\times\mathcal{A}_G$ classification scheme for MBL Floquet drives in 1d with finite abelian on-site symmetry group $G$, and full SPT eigenstate order. Here Cl$_G$ is the undriven SPT classification, and $\mathcal{A}_G$ are the set of 1d or twisted 1d representations of $G$ defined in \eqnref{eq:twistedrep}. Only certain many-body localizable\cite{Potter15} SPT eigenstate orders are expected to persist in the Floquet setting\cite{Khemani15}. For this reason we restrict attention to SPT orders with finite $G$, which Ref.~\onlinecite{Potter15} suggests are many-body localizable. We further restrict to abelian $G$ for simplicity. $\zzt^{\text{fp}},\zzt^{\mathscr{T}}$ are the fermion parity, and $\mathscr{T}^2=1$ time reversal symmetry groups respectively.}\label{tab:class}
\end{table}

\section{An algebraic characterization}\label{s:algebraic}
Having argued that the pumped charge is a robust property of symmetric Floquet unitaries with eigenstate order, we now show how the presence of the pumped charge affect the spectra of Floquet unitaries. In the process we condense our results into a concise algebraic formalism. We do this by extending the formalism of Ref.~\onlinecite{Kitaev11}, which was developed to deal with equilibrium SPT states. In this undriven setting, Ref.~\onlinecite{Kitaev11} starts by considering the nearly degenerated ground states of an SPT on an open 1d chain $S=[L,R]$. These ground states are indistinguishable in the bulk, and differ only near $L,R$. The global symmetry group acts on the this low energy space. By locality, for an extensively large system, the global symmetry must act like $V(g) \rightarrow \widehat{g}_L \widehat{g}_R$ within this space, where $\widehat{g}_{L,R}$ localized near the $L,R$ end of the system respectively. Note that $\widehat{g}_{L,R}$ are only defined up to phases, and indeed need only obey $\widehat{g}_{L}\widehat{h}_{L}=c(g,h)\widehat{gh}_{L}$ where $c$ is a $2$-cocycle defining a projective representation of $G$ (similar for the right edge). To consider anti-unitary symmetry groups, it is helpful to define a homomorphism $\alpha: G \rightarrow \mathbb{Z}_2$ where $\alpha(g)=\pm 1$ for $g$ unitary/anti-unitary respectively. With this in mind the associativity of the $G$ action on say the left edge leads to a relation
\be
c(g,h)^{\alpha(f)} c(fg,h)^{-1} c(f,gh) c(f,g)^{-1}=1.
\ee
which is the defining relation for a 2-cocycle.  Ref.~\onlinecite{Kitaev11} then argue that this 2-cocycle is the relevant datum identifying the SPT
in question. Within the low energy subspace, in the thermodynamic limit, $H$ must act like a scalar, and it can be argued that $\hat{g}_L, \hat{g}_R$ individually commute with the Hamiltonian. Hence, the low energy subspace is some representation of the algebra generated by the operators $\{\hat{g}_L, \hat{g}_R\}$. The form of this algebra is determined entirely by the choice of cocycle $c$\footnote{As well as the fermion parity of the $\hat{g}$ in fermionic systems.}. This is the classification of $1$D SPTs in brief. For system's whose entire spectrum is MBL with SPT eigenstate order, the above statements hold not only for the ground state subspace, but for a complete set of degenerate multiplets of excited eigenstates. 
%
%
%
%

Consider now a Floquet SPT drive on an open chain with full eigenstate order. We can play a similar
game. Again we can fix the bulk eigenstate (i.e., bulk conserved quantities) and consider the symmetry action in this restricted subspace. Again, when we consider the symmetry acting on a particular edge, we find that it is characterized by some
2-cocyle $\widehat{g}\widehat{h}=c(g,h)\widehat{gh}$. However, in the Floquet
case there is potentially another datum determining the edge structure. In particular, while previously $\widehat{g}_L$ commuted individually with the Hamiltonian, in the Floquet case we see that $\widehat{g}_{L,R}$ need not necessarily commute with $U_f$. To see this, let us treat unitary and anti-unitary $G$ separately. 

In the previous section we saw that the Floquet unitary takes the form $U_f=v_L v_R e^{- i f}$. Fixing the bulk state -- i.e., the conserved quantities in the bulk --  the Floquet unitary acts like $U_f \propto v_L v_R$. We argued in \lemref{lem:phase} that the global symmetry commutes with $v_L,v_R$ individually up to a phase characterized by $\kappa_L(g)=[V(g):v_L]$ .  This quantity, in turn, determines the commutation between $\widehat{g}_L,\widehat{g}_R$ and $U_f$. So, the algebra of symmetry operators in the edge space are characterized by a 2-cocycle $c$ and  $\kappa_L(g)=[V(g):v_L]$, which defines a 1d representation of the gauge group. For fermionic systems one should bear in mind the possibility that operators on distant edges may anti-commute if they are fermion parity odd.

Recall that in the anti-unitary symmetry group case we call the Floquet unitary symmetric if $i \log U_f$ can be chosen to be a $G$ symmetric Hamiltonian. This means that $V(g) U_f V^{-1}(g)= U^{\alpha(g)}_f$. As before, consider the action of $U_f$ into the edge subspace, which again by locality takes form $v_L v_R$. The global symmetry in this case should obey $(V(g):v_L v_R)=1$ (see \eqnref{eq:groupcomm}) -- although we were unable to prove this in the anti-unitary case. Assuming we can, we have $(V(g):v_{L,R})=e^{i \theta_{L,R}} 1$ -- this information is captured by just one quantity $\kappa_L(g)=(V(g):v_L)$. This $\kappa_L$ quantity in turn determines the  $U_f$ commutation relations with the edge symmetry operators $\widehat{g}_L,\widehat{g}_R$. So the symmetry algebra at the edge is again characterized by a 2-cocycle and a U$(1)$ phase $\kappa_L(g)=V(g) v_L V(g)^{-1} v^{-\alpha(g)}_L$.  When the symmetry group is anti-unitary, this phase does not quite form a 1d representation as in the unitary case. Instead it obeys \eqnref{eq:twistedrep} hence the data determining the drive are $c(g,h), \kappa(g)$ where $\kappa(g)$ is a kind of twisted 1d representation, the set of which we denote $\mathcal{A}_G$.

Having proposed a classification for 1d Floquet SPT drives, in the next two sections we describe some instructive examples. First in \secref{s:Z2xZ2} we look at an interacting bosonic Floquet SPT drive with $G=\mathbb{Z}_2\times \mathbb{Z}_2$. Then in \secref{s:BDI} we look at two examples of drives with anti-unitary symmetry groups of form $G=  \mathbb{Z}_2^{\mathscr{T}}, \mathbb{Z}_2^{\mathscr{T}}\times \mathbb{Z}_2^{\text{fp}}$ -- the latter is of particular interest, as it corresponds to an interacting version of the fermionic BDI symmetry class. In all cases, we will provide explicit examples of drives within each of the proposed Floquet phases.

\section{Edge structure for $G=\bbZ_{2}\times \bbZ_{2}$ Floquet drives}\label{s:Z2xZ2}
In this section we focus on bosonic paramagnets with unbroken global $G=\bbZ_{2}\times \bbZ_{2}$ symmetry. This example is interesting because it involves an intrinsically interacting bosonic system, and it breaks the Cl$\times$Cl classification pattern seen in the classification of non-interacting fermionic Floquet drives\cite{Nathan15}. To wit, $G=\bbZ_{2}\times \bbZ_{2}$ has an undriven SPT classification of Cl$=\mathbb{Z}_2$, corresponding to a trivial paramagnet and non-trivial SPT state, while we predict in \secref{s:generalities} a $\mathbb{Z}_2\times \mathbb{Z}_2\times \mathbb{Z}_2 (\neq \text{Cl}\times \text{Cl})$ classification. After describing the SPT order in the undriven setting, we provided examples of drives in each of the eight putative phases in the conjectured classification. We describe the edge theory for some of these drives using the formalism of \secref{s:algebraic}. 

Consider a chain with on-site local Hilbert space ${\mid g^{1},g^{2}\rangle}$ where $g^{1},g^{2}\in\mathbb{Z}_{2}$. Let $Z^1,Z^2$ be the operators measuring $g^{1},g^{2}$, and let $X^{1}, X^{2}$ act as ${X^{1}\mid g^{1},g^{2}\rangle=\mid-g^{1},g^{2}\rangle}$
and ${X_{r}^{2}\mid g^{1},g^{2}\rangle}={\mid g^{1},-g^{2}\rangle}$
respectively. It is clear that $X^{1,2},Z^{1,2}$ behave like $\sigma^{x},\sigma^{z}$ Pauli-matrices. States with global $\zzt\times\zzt$ symmetry generators $\prod_r X^{1,2}_{r}$ have a $H^{2}\left(G,\text{U}\left(1\right)\right)=\mathbb{Z}_{2}$ classification from group cohomology -- hence there are two SPT fixed points, corresponding to the trivial paramagnet and SPT 

\begin{align}
H_{0} & =-\sum_{r=1}^{N}\left(h_{r}^{1}X_{r}^{1}+h_{r}^{2}X_{r}^{2}\right)\\
H_{1} & =-\sum_{r=2}^{N} h_{r}^{1}X_{r}^{1}Z_{r-1}^{2}Z_{r}^{2}- \sum_{r=1}^{N-1}h_{r}^{2}X_{r}^{2}Z_{r}^{1}Z_{r+1}^{1}\label{eq:z2xz2ham}
\end{align}
respectively. Both model Hamiltonians are sums of commuting operators. The l-bits for the trivial paramagnet are $\{X^{1}_r,X^{2}_r\}$ while those for the SPT are $\{X^{1}_r Z_{r-1}^{2} Z_{r}^{2},X^{2}_rZ_{r}^{1}Z_{r+1}^{1}\}$. As alluded to in \secref{s:algebraic}, we can fix the `bulk' conserved l-bits appearing in \eqnref{eq:z2xz2ham}, and ask how the global symmetry transformations act on the residual edge degrees of freedom $V(g)\rightarrow \widehat{g}_L \widehat{g}_R$ -- the forms of the generators projected onto this subspace with fixed l-bits is summarized in \tabref{tab:z2z2proj}. 

\begin{table}
\begin{tabular}{|c|c|c|c|c|}
\hline 
{Generators\textbackslash{}Bulk order} & \multicolumn{2}{c|}{$\text{\textbf{triv. pm}}$} & \multicolumn{2}{c|}{\textbf{spt}}\tabularnewline
\cline{2-5} 
 & L & R & L & R\tabularnewline
\hline 
$\widehat{(-1,1)}$ & $X_{1}^{1}$ & $X_{N}^{1}$ & $X_{1}^{1}Z_{1}^{2}$ & $Z_{N}^{2}$\tabularnewline
\hline 
$\widehat{(1,-1)}$ & $X_{1}^{2}$ & $X_{N}^{2}$ & $Z_{1}^{1}$ & $X_{N}^{2}Z_{N}^{1}$\tabularnewline
\hline 
\end{tabular}
\caption{This table summarizes the form of the global symmetry $\mathbb{Z}_2\times \mathbb{Z}_2$ generators projected onto a subspace with fixed bulk conserved quantities $V(g)\rightarrow \widehat{g}_L \widehat{g}_R$. In the SPT case, the generators at a particular edge ($L,R$) furnish a projective representation. The group $\mathbb{Z}_2\times \mathbb{Z}_2$ is presented as a multiplicative group on set $\{(\pm 1, \pm 1),(\pm 1, \mp 1)\}$.}\label{tab:z2z2proj}
\end{table}%

\subsection{MBL Binary drives realizing the Floquet phases}
Here we construct examples of drives for the eight putative Floquet phases. We use a three part drive of the form
\be\label{eq:trinary0}
U_{f}= e^{-it_{2} K_{2}} e^{-i t_{1} K_{1} }e^{-it_{0} K_{0}}\punc{,}
\ee
where $K_0$ is one of $H_0,H_1$, while $K_1, K_2$ are chosen from

\begin{align*}
H_{\text{FM}_1} & =-\sum_{r=1}^{N} Z^1_r Z^1_{r+1} \\
H_{\text{FM}_2} & =-\sum_{r=1}^{N} Z^2_r Z^2_{r+1}\punc{.}
\end{align*}

In $H_{0,1}$, we choose $h^{1,2}_{r}$ disordered with mean $1$. In all of the eight examples, we always set $t_{1,2} = 0$ or $\pi/2$. For these choices of $t_1,t_2$, we can use the identity $e^{- i \frac{\pi}{2} H_{\text{FM}_{1,2}}} \propto Z^{1,2}_{1} Z^{1,2}_N$  to show that
$$
U_{f}=v_L v_R e^{-it_{0} K_{0}}
$$ 
where (for instance) $v_{L}$ looks like one of 
$$
1,Z^{1}_1,Z^{2}_1,Z^{1}_1 Z^{2}_1
$$ 
The two possible choices of $K_{0}$ (trivial PM or SPT), along with the four choices of $v_L$ above, give the eight elements of the classification $\mathbb{Z}_2\times \mathbb{Z}_2 \times \mathbb{Z}_2$. By construction, the drive in question has a complete set of exactly local bulk conserved quantities (from $K_0$), and with some minor local symmetric changes of basis at the edge (similar to those below \eqnref{eq:classDlocun}) $v_{L,R}$ can be chosen to commute with $K_0$. Hence, the drive constructed gives a `fixed point' realization of the different Floquet classes predicted by our framework in \secref{s:generalities}. In the remainder of the section, we examine a selection of the eight constructed drives, and explain the structure of their eigenspectra. A detailed discussion of the edge states for all eight cases can be found in \appref{app:Z2Z2full}. 

\subsubsection{Undriven example: $t_{1,2}=0$}
In these cases the unitary is just
$$
U_{f}=e^{-it_{0} K_{0}}\punc{,}
$$ 
so the spectrum of $U_f$ is just the spectrum of $K_0$. Fixing the bulk l-bits we know that the symmetry action factorizes as $V(g)\rightarrow \widehat{g}_L \widehat{g}_R$, and the exponentially degenerate eigenspaces form representations of the algebra generated by $\widehat{g}_L, \widehat{g}_R$.

For $K_0$ trivial paramagnetic for instance, the states need only form a representation of the algebra generated by $\{ X^1_1,  X^2_1, X^1_N,  X^2_N \}$. As all of the elements of this algebra commute, representations of this algebra can be 1 dimensional. In physical terms there are no protected degeneracies at the edge of this system\cite{Kitaev11}.

On the other hand, for $K_0= H_1$, $\widehat{g}_L$ furnish a projective representation of the symmetry group -- those listed in \tabref{tab:z2z2proj}. Fixing the bulk integrals of motion, the residual edge degrees of freedom form some representation of the algebra generated by $\{\widehat{g}_L , \widehat{g}_R \}$ for all $g\in G$. In the present case the symmetry generators at a particular edge do not generally commute. This non-commutation of the symmetry generators implies that there is at least a protected a two-fold degeneracy associated with each edge -- indeed, in the present example, there is exactly a two-fold degeneracy associated with each edge. So the spectrum of the Floquet unitary $U_f$ has a spectral pairing, with every eigenstate being a part of a pair at the same quasi-energy.

\subsubsection{Non-trivial Floquet example}
\begin{table}
\caption{This table shows the eigenspectrum structure at the edge of the $\mathbb{Z}_2\times \mathbb{Z}_2$ SPT for a trivial bulk paramagnetic order, and $v_L = Z^1_1 Z^2_1$ i.e., $\kappa(g^1,g^2)=g^1 g^2$.}
\begin{tabular}{|c|c|c|}
\hline 
 & $U_f \propto v_{L}v_{R}$ & $ X^1_1  X^1_N $\tabularnewline
\hline 
$\mid\psi\rangle $ & $u$ & $u'$\tabularnewline
\hline 
$v_{L}\mid\psi\rangle $ & $u$ & $-u'$\tabularnewline
\hline 
$X^1_1   \mid\psi\rangle $ & $-u$ & $u'$\tabularnewline
\hline 
$X^1_1  v_{L} \mid\psi\rangle $ & $-u$ & $-u'$\tabularnewline
\hline 
\end{tabular}\label{tab:repthryexampleZ2xZ2}
\end{table}
In this example we set $t_{1,2}=\pi/2$ and $K_{1,2}=H_{\text{FM}_{1,2}}$ respectively, and $K_0 = H_0$. The resulting Floquet unitary is of the form
$$
U_f=v_L v_R  e^{- i t_0 H_0}
$$
Using a local symmetric change of basis we can rewrite this as 
$$
U_f=v_L v_R  e^{- i t_0 \sum^{N-1}_{s=2} h^1_s X^1_s + h^2_s X^2_s} 
$$
with $v_L= Z^1_1 Z^2_1$ and $v_R=Z^1_N Z^2_N$, which commute with the bulk l-bits $X^{1,2}_s$ for $s=2,\ldots, N-1$. Fixing the bulk l-bits, the Floquet unitary acts like 
$$
\propto v_L v_R
$$
on the edge degrees of freedom.  Looking at the commutation relations between the global symmetry generators and $v_L$ gives pumped charge $\kappa(g)= g^1 g^2$. As the bulk order is trivial paramagnetic, the symmetry operators at the edge look like $\{ X^1_1,  X^2_1, X^1_N,  X^2_N \}$ from \tabref{tab:z2z2proj}. Having fixed all the bulk l-bits, the remaining edge degrees of freedom form a representation of an algebra 
$$
\mathfrak{a}= \text{gen}\{X^1_1,  X^2_1, X^1_N,  X^2_N, v_L,v_R\}.
$$
As a result, we can show that the possible edge states come in quadruplets (two protected degrees of freedom at each edge). These quadruplets do not all lie at the same quasi-energy as was the case in the undriven case, but the quasi-energy spacings are protected. To see this, first note that $\mathfrak{a}$ has a centre $Z(\mathfrak{a})$ generated by $\{X^1_1  X^2_1, X^1_N  X^2_N\}$. As these operators commute with all of $\mathfrak{a}$ their eigenvalues can be fixed. This amounts to modding out the centre and considering the representations of the algebra $\mathfrak{a}'=\mathfrak{a}/Z(\mathfrak{a})\sim\text{gen}\{v_L,v_R,X^1_1,X^1_N\}$. 

To establish the representations of $\mathfrak{a}'$ it is helpful to (following Ref.~\onlinecite{Kitaev03}) identify a maximal commuting sub-algebra $\mathfrak{b}=\text{gen}\{v_L v_R, X^1_1 X^1_N\}$. We proceed by picking  a simultaneous eigenvector $\mid\psi\rangle$ of the generators of $\mathfrak{b}$ --  denote the corresponding eigenvalues $u,u'$ respectively. Consider acting on this state with the remaining elements of the algebra. We get at least four states with different eigenvalues in $\mathfrak{b}$. \tabref{tab:repthryexampleZ2xZ2} shows the distinct states arising from this procedure. We also record the eigenvalues of $U_f $ which simply act like $v_Lv_R$ in this eigenspace. In summary, fixing the bulk l-bits, we find there are four edge states (two at each edge)  spread evenly between two $U_f$ eigenvalues $u,-u$. Within each  $U_f$ eigenspace, there are two eigenstates distinguished by the $X^1_1  X^1_N$ eigenvalue, for instance. 
%
%

\begin{table*}[t]
\begin{tabular}{|c|c|c|c|c|c|c|c|c|}
\hline 
{Params.\textbackslash{}Class.} & $\text{\textbf{pm}},(1,1)$ & $\text{\textbf{pm}},(-1,1)$ & $\text{\textbf{pm}},(1,-1)$ & $\text{\textbf{pm}},(-1,-1)$ & $\text{\textbf{spt}},(1,1)$ & $\text{\textbf{spt}},(-1,1)$ & $\text{\textbf{spt}},(1,-1)$ & $\text{\textbf{spt}},(-1,-1)$\tabularnewline
 & $v_{L}=1$ & $v_{L}=Z^{1}$ & $v_{L}=Z^{2}$ & $v_{L}=Z^{1}Z^{2}$ & $v_{L}=1$ & $v_{L}=Z^{1}$ & $v_{L}=Z^{2}$ & $v_{L}=Z^{1}Z^{2}$\tabularnewline
\hline 
$(K_{0},t_{0})$ & $(H_{0},1)$ & $(H_{0},1)$ & $(H_{0},1)$ & $(H_{0},1)$ & $(H_{1},1)$ & $(H_{1},1)$ & $(H_{1},1)$ & $(H_{1},1)$\tabularnewline
\hline 
$(K_{1},t_{1})$ & - & $(H_{\text{FM}_{1}},\frac{\pi}{2})$ & $(H_{\text{FM}_{2}},\frac{\pi}{2})$ & $(H_{\text{FM}_{1}},\frac{\pi}{2})$ & - & $(H_{\text{FM}_{1}},\frac{\pi}{2})$ & $(H_{\text{FM}_{2}},\frac{\pi}{2})$ & $(H_{\text{FM}_{1}},\frac{\pi}{2})$\tabularnewline
\hline 
$(K_{2},t_{2})$ & - & - & - & $(H_{\text{FM}_{2}},\frac{\pi}{2})$ & - & - & - & $(H_{\text{FM}_{2}},\frac{\pi}{2})$\tabularnewline
\hline 
\end{tabular}
\caption{This table shows how to construct a drive \eqnref{eq:trinary0} with symmetry group $\mathbb{Z}_2 \times\mathbb{Z}$, corresponding to a prescribed Floquet phase. The Floquet phases are labelled by bulk eigenstate order (trivial paramagnet \textbf{pm} or SPT \textbf{spt}), and a pumped charge $\kappa$ which is uniquely determined by pair $(\kappa(-1,1),\kappa(1,-1))$. }\label{tab:z2z2drives}
\end{table*}%

\section{Anti-unitary examples}\label{s:BDI}
\begin{table}
\begin{tabular}{|c|c|c|c|c|}
\hline 
Drive\textbackslash{}$(\hat{\mathscr{T}}^{2},\eta)$ & $(1,1)$ & $(-1,1)$ & $(1,-1)$ & $(-1,-1)$\tabularnewline
\hline 
$(K_{0},t_{0})$ & $(H_{0},1)$ & $(H_{1},1)$ & $(H_{\text{FM}},1)$ & $(H_{\text{FM}},1)$\tabularnewline
\hline 
$(K_{1},t_{1})$ & - & - & $(H_{0},1)$ & $(H_{1},1)$\tabularnewline
\hline 
\end{tabular}
\caption{This table summarizes the drive parameters for \eqnref{eq:trinary} and their corresponding Floquet order. The Floquet order is summarized by $(\hat{\mathscr{T}}^{2},\eta)$ -- the first argument characterizes the bulk order by specifying the projective representation of time reversal at the edge, while $\eta$ captures how $\hat{\mathscr{T}}$ fails to commute with $U_f$.  }
\label{tab:z2T}
\end{table}

\begin{table}
\begin{tabular}{|c|c|c|c|c|}
\hline 
{Drive\textbackslash{}$(x,y,z)$} & $(x,1,1)$ & $(x,-1,1)$ & $(x,1,-1)$ & $(x,-1,-1)$\tabularnewline
 & $v_{L}=1$ & $v_{L}=\bar{\psi}_{1}$ & $v_{L}=Z_{1}$ & $v_{L}=Z_{1}\bar{\psi}_{1}$\tabularnewline
\hline 
$(K_{0},t_{0})$ & - & $(H^{B},\pi/2)$ & $(H^{A},\pi/2)$ & $(H^{A},\pi/2)$\tabularnewline
\hline 
$(K_{1},t_{1})$ & - & $(H^{C},\pi/2)$ & - & $(H^{B},\pi/2)$\tabularnewline
\hline 
$(K_{2},t_{2})$ & - & - & - & $(H^{C},\pi/2)$\tabularnewline
\hline 
$(K_{3},t_{3})$ & $(H_{x},\pi/4)$ & $(H_{x},\pi/4)$ & $(H_{x},\pi/4)$ & $(H_{x},\pi/4)$\tabularnewline
\hline 
\end{tabular}
\caption{BDI drives are conjectured to by classified by $(x,y,z)\in \mathbb{Z}_8\times \mathbb{Z}_2 \times \mathbb{Z}_2$. For a fixed bulk order $x$, the table gives drive parameters corresponding to any of the four choices of $(y,z)$, where $y=\kappa(P),z=\kappa(\mathscr{T})$ (see \secref{sss:BDIIsing}). The form of the corresponding edge unitaries $v_L$ is also shown.}\label{tab:BDIdrives}
\end{table}
In this section we grapple with SPT drives with time reversal symmetry. As such SPTs are not as well understood in the context of MBL and eigenstate order\cite{Potter15} our results in this section are more tentative, and based on the heuristic arguments in \secref{ss:antiunitary} and \secref{s:algebraic}. In this section, we will deal with two examples of SPTs with abelian symmetry groups with time reversal, namely $G=\mathbb{Z}^{\mathscr{T}}_2,\mathbb{Z}^{\mathscr{T}}_2\times \mathbb{Z}^{\text{fp}}_2$. We check that the arguments of \secref{ss:antiunitary} certainly apply to these two cases, so that the classifications are of form $\mathbb{Z}_2,\mathbb{Z}_8\times\mathbb{Z}_2\times \mathbb{Z}_2$ respectively as predicted, and give examples of drives which should fall into each putative Floquet phase. We highlight in particular the $G=\mathbb{Z}^{\mathscr{T}}_2\times \mathbb{Z}^{\text{fp}}_2$  case, which is a Fermionic system with time reversal symmetry (called `BDI' in the free fermion context). Our results show that the free fermion classification of the BDI Floquet classes breaks down from $\mathbb{Z}\times \mathbb{Z}$ (see Refs.~\onlinecite{Nathan15,Roy15}) to  $\mathbb{Z}_8\times \mathbb{Z}_2\times  \mathbb{Z}_2$ in the presence of interaction  in a manner similar to that seen in the undriven setting\cite{Kitaev11}.

\subsection{$G=\mathbb{Z}^{\mathscr{T}}_2$, bosonic system}\label{s:z2TI}
First consider a bulk SPT phase with just $\mathbb{Z}^{\mathscr{T}}_2$ symmetry and $\mathscr{T}^2=1$. We have not found an explicit discussion of such phases in the literature, although the relevant cohomology calculation is found in Ref.~\onlinecite{Chen11}. First construct the undriven SPT states. Consider a system with an Ising $Z_r=\pm1$ degree of freedom on each site, and a symmetry $\prod_r X_r K$ where $K$ is complex conjugation. An example of such a Hamiltonian
\be\label{eq:H0Z2T}
H_0 = -\sum_r  h_r X_r 
\ee
has paramagnetic order and no symmetry protected edge state. On the other hand
\be\label{eq:H1Z2T}
H_1 = -\sum_r  h_r X_r Z_{r-1} Z_{r+1}  
\ee
has edge states which transform according to a projective representation $\hat{\mathscr{T}}_L^2 = -1$\footnote{Note that a clean variant of \eqnref{eq:H1Z2T} emerges in the high frequency expansion of the interacting drives considered in Ref.~\onlinecite{Iadecola15}.}. Note that $H_0,H_1$ are commuting stabilizer Hamiltonians, and the local integrals of motion $X_r, X_r Z_{r-1} Z_{r+1}$ commute with $\mathscr{T}$ while taking values $\pm 1$, so that the arguments of \secref{ss:antiunitary} apply.  According to that discussion, and to the discussion in \secref{s:algebraic}, there will be just two possible Floquet phases for a given bulk order, distinguished by the pumped charges $\kappa(\mathscr{T})=\pm 1$. This pumped charge in turn determines the commutation relations between $\hat{\mathscr{T}}_L$ and the Floquet unitary restricted to the edge subspace $(\hat{\mathscr{T}}_L:v_L)=\pm 1$.  

Here we claim to construct examples of the four Floquet pumps using trinary drives. It is useful to define an auxiliary ferromagnetic drive

\be\label{eq:FMZ2T}
H_{\text{FM}} = -\sum_r  Z_{r} Z_{r+1}  .
\ee
Consider Floquet pumps of the form
\be\label{eq:trinary}
U_{f}=e^{-i\frac{t_{0}}{2}K_{0} }e^{-it_{1} K_{1}} e^{-i\frac{t_{0}}{2}K_{0} }.
\ee
All we need to do is specify $t_0,t_1$ and $K_{0,1}$ are either \eqnref{eq:H0Z2T} or \eqnref{eq:H1Z2T}, where we choose $h_r$ to be say log-normal distributed with mean $1$. The possible classes of drives will be labelled by $(\hat{\mathscr{T}}^2,\eta)$ where $\hat{\mathscr{T}}^2=\pm 1$ determine the SPT cocycle and hence the bulk order, while $\eta = \kappa(\mathscr{T})= (\mathscr{T}:v_L)$ as discussed above. \tabref{tab:z2T} summarizes which Hamiltonians need to be chosen for a given Floquet phase $(\hat{\mathscr{T}}^2,\eta)$. 

\subsection{Interacting BDI drives i.e., $G=\mathbb{Z}^{\text{fp}}_2\times \mathbb{Z}^{\mathscr{T}}_2$ }
Here we put forward a tentative classification of interacting 1d `Class BDI' MBL Floquet drives. The symmetry group is $G=\mathbb{Z}^{\text{fp}}_2\times \mathbb{Z}^{\mathscr{T}}_2$ where time reversal obeys $\mathscr{T}^{2}=1$ on the fundamental fermions. The undriven problem has a $\mathbb{Z}_{8}$ classification\cite{Kitaev11}. From \secref{s:algebraic}, we expect the Floquet drives to have a $\mathbb{Z}_{8}\times\mathbb{Z}_{2} \times \mathbb{Z}_{2}$ classification, in contrast to the $\mathbb{Z}\times\mathbb{Z}$ classification found in the non-interaction band theory picture \cite{Nathan15,Gannot15}.  We now attempt to explain this collapse in classification using some example drives. First in \secref{sss:stacking} we consider stacking a number of the (clean) class D drives considered in \secref{s:motivatingexample}. Then in \secref{sss:BDIIsing} we use another realization of the same SPT order and Floquet phase involving a single Majorana chain coupled to additional Ising degrees of freedom.

\subsubsection{Stacking argument}\label{sss:stacking}
In this section we get a more concrete feel for how the $\mathcal{A}=\mathbb{Z}_{2}\times\mathbb{Z}_{2}$
part of the Floquet classification comes about by stacking many class D time reversal symmetric drives. While the stacked models we consider will have many extraneous bulk degrees of freedom, it allows us to extract useful intuition. Consider a drive with
$k$ Kitaev Majorana chains, and with net Floquet unitary

\be\label{eq:UfBDI}
U_{f}=e^{-i\frac{t_{0}^{\alpha}}{2}H_{0}^{(\alpha)}}e^{-it_{1}^{\alpha}H_{1}^{(\alpha)}}e^{-i\frac{t_{0}^{\alpha}}{2}H_{0}^{(\alpha)}}
\ee
 $\alpha=1,\ldots,k$ label chains and

\begin{align*}
H_{0}^{(\alpha)}= & -\sum_{r=1}^{N}i\bar{\psi}_{r}^{\alpha}\psi_{r}^{\alpha}\\
H_{1}^{(\alpha)}= & -\sum_{r=1}^{N-1}i\psi_{r}^{\alpha}\bar{\psi}_{r+1}^{\alpha}.
\end{align*}

Note that the local conserved quantities in each of these Hamiltonians take values in $\pm 1$ and commute with time reversal symmetry as per the requirements of \secref{ss:antiunitary}. The Majoranas are such that $ T\psi^{\alpha}_r T=\psi^{\alpha}_r$ and $\mathscr{T}\psib^{\alpha}_r T=-\psib^{\alpha}_r$.

 To obtain a drive with $n_{L}^{0}$ zero quasi-energy Majoranas and
$n_{L}^{\pi}$ $\pi$ quasi-energy Majoranas at (say) the left edge,
set $k=n_{L}^{0}+n_{L}^{\pi}$ and set $t_{1}^{\alpha=1,\ldots,n_{L}^{\pi}}=\frac{\pi}{2}+\epsilon$,
$t_{0}^{\alpha=1,\ldots,n_{L}^{\pi}}=\frac{\pi}{2}$ and $t_{1}^{\alpha=n_{L}^{\pi}+1,\ldots,k}=\epsilon$
and $t_{0}^{\alpha=n_{L}^{\pi}+1,\ldots,k}=0$ where say $0<\epsilon<1$
(the specific value is unimportant). The resulting Floquet unitary
is

\begin{align*}
U_{f} & =e^{-i\frac{t_{1}^{\alpha}}{2}H_{1}^{(\alpha)}}e^{-it_{0}^{\alpha}H_{0}^{(\alpha)}}e^{-i\frac{t_{1}^{\alpha}}{2}H_{1}^{(\alpha)}}\\
 & =e^{-i\frac{t_{1}^{\alpha}}{2}H_{1}^{(\alpha)}}\prod_{\alpha>n_{L}^{\pi}}P_{\alpha}e^{-i\frac{t_{1}^{\alpha}}{2}H_{1}^{(\alpha)}}\\
 & =\prod_{\alpha>n_{L}^{\pi}}P_{\alpha}e^{-it_{1}^{\alpha}H_{1}^{(\alpha)}}\\
 & \propto\prod_{\alpha<n_{L}^{\pi}}\bar{\psi}_{1}^{\alpha}\prod_{\alpha<n_{L}^{\pi}}\psi_{N}^{\alpha}e^{-i\epsilon\sum_{\alpha}H_{1}^{(\alpha)}}
\end{align*}
Note that $U_{f} \bar{\psi}_{1}^{\alpha} U_{f}^{-1}=(-1)^{(\alpha<n_{L}^{\pi})}\bar{\psi}_{1}^{\alpha}$
so that $\bar{\psi}_{1}^{\alpha=1,\ldots,n_{L}^{\pi}}$ are $\pi$
Majoranas and $\bar{\psi}_{1}^{\alpha=n_{L}^{\pi}+1,\ldots,k}$ are
zero Majoranas.  

Consider a drive with bulk classification $m\in\mathbb{Z}_{8}=\{0,1,2,\ldots7\}$.
By the eightfold Kitaev-Fidkowski classification, we may as well choose
the above drive with any $k=m+8n$ -- it is convenient for our purposes
to choose $k=m+8$ so that there are at least $8$ chains present.
We will find that the properties of $v_{L}$ as a function of Majoranas comprising $v_L$, namely $l\equiv n_{L}^{\pi}$, depend only on $l$ modulo $4$. 

First note that if $l=4$, $v_{L}=\bar{\psi}_{1}^{1}\bar{\psi}_{1}^{2}\bar{\psi}_{1}^{3}\bar{\psi}_{1}^{4}=:\bar{\psi}_{1}^{(1234)}$. Note this term can be removed from $U_{f}$ by extending the old Floquet drive by a local term
\be\label{eq:fourfold}
U_f \rightarrow U'_f=e^{-\frac{i\pi}{4}t\bar{\psi}_{1}^{(1234)}} U_f e^{-\frac{i\pi}{4}t\bar{\psi}_{1}^{(1234)}}
\ee
The resulting drive $U'$ is still time reversal invariant, and in particular $(V(\mathscr{T}):U'_f)=1$. The modification also respects local parity symmetric. The same thing can be done at the right hand edge. Hence, we have locally and symmetrically modified the Floquet drive to obtain

\[
U'_{f}=e^{-i\epsilon\sum_{\alpha=1}^{m+8}H_{1}^{(\alpha)}}.
\]
This Floquet unitary can clearly be engineered using a time independent
Hamiltonian drive with bulk $m+8$ eigenstate order, on a system
with boundary. Hence, when it comes to robust eigenstate properties, the
$l=4$ drive should be considered the same as the $l=0$ drive with the same bulk order. Using
this style of argument, it readily follows that the robust physical
properties of a drive of form \eqnref{eq:UfBDI} with any $m$ should depend only on $l\mod4$. 

Let us now show that $l=0,1,2,3$ have distinct physical properties.
Note first that $v_{L}\propto\bar{\psi}^{(12)}$ cannot be removed
in the above manner. This follows from $(V(\mathscr{T}):\bar{\psi}_{1}^{(12)})=(V(\mathscr{T}):\psi_{N}^{(12)})=-1$.
If we modify $U_f$ using some unitary $w_{L}$ acting on the
left edge Majoranas, then the new Floquet unitary $U'_{f}$ must have
the time reversal property
\begin{align}
1&=(V(\mathscr{T}):U'_{f})\nonumber\\
&=V(\mathscr{T})v_{L}v_{R}w_{L}V(\mathscr{T})^{-1}v_{L}v_{R}w_{L}\nonumber\\
&=-V(\mathscr{T})v_{L}w_{L}V(\mathscr{T})^{-1}v_{L}w_{L}\label{eq:BDIinter}.
\end{align}
where the last equality follows from the fact that $v_R$ has the same time reversal property as $v_L$, namely  $(V(\mathscr{T}): v_{L,R})=-1$. Were it possible to completely remove $v_{L}$ with such a $w_{L}$, then $V(\mathscr{T})v_{L}w_{L}V(\mathscr{T})^{-1}v_{L}w_{L}=1$. But this is inconsistent with \eqnref{eq:BDIinter}. 

The drives with $v_{L}=\bar{\psi}^{(1)}$ are clearly non-trivial
(from the class D part of the paper) because $v_{L}$ is fermion parity
odd. What remains however, is to show that $v_{L}=\bar{\psi}^{(1)}$
is distinct from $v_{L}=\bar{\psi}^{(123)}$. In fact this follows
readily using the above method: Note that $(V(\mathscr{T}):\bar{\psi}^{(1)})=-1$
while $(V(\mathscr{T}):\bar{\psi}^{(123)})=1$. The different time reversal properties
of these potential $v_{L}$ mean they cannot be locally tuned to one
another while preserving time reversal invariance.  In summary, fixing bulk class $m$, there appear to be four distinct drives $l=0,1,2,3$ labelled by the four combinations of $\mathbb{Z}_{2}$ numbers $(V(P):v_{L}),(V(\mathscr{T}):v_{L})=\pm1$. Thus the different Floquet drives appear to be labelled by elements of $\mathbb{Z}_8\times \mathbb{Z}_2 \times \mathbb{Z}_2$. Note that as sets (though not as groups) $\mathbb{Z}_2\times \mathbb{Z}_2$ is equivalent to $\mathbb{Z}_4$, so we could also say that the Floquet phases lie in the set $\mathbb{Z}_8\times \mathbb{Z}_4$. This latter presentation is preferable if one wishes the classification group to reflect the fourfold (see \eqnref{eq:fourfold}) manner in which the pumped charge changes as we stack multiple Floquet systems atop one another. In other words one can view different Floquet phases as forming an abelian group  $\mathbb{Z}_8\oplus \mathbb{Z}_4$, with addition corresponding to taking a tensor product of systems.  

\subsubsection{Alternative setup}\label{sss:BDIIsing}
Consider a chain with and onsite Hilbert space consisting of $\psi,\psib$ majoranas, as well as a $\mathbb{Z}_2$ degree of freedom $Z_r=\pm1$ where $Z_r$ a Pauli-matrix. Let time reversal act like $\mathscr{T}=\prod_{r} X_r K$ where $K$ is complex conjugation. Then the Majorana fermions have the usual time reversal transformations, and $\mathscr{T}^2=1$. It can be verified (although we have not found an appropriate reference) that the following Hamiltonians capture the eight possible MBL phases of Fidkowski and Kitaev\cite{Kitaev11} with $G=\mathbb{Z}^{\text{fp}}_2\times \mathbb{Z}^{\mathscr{T}}_2$ symmetry
\begin{align}
H_0 &= \sum_r h^{(0)}_r i\psi_r\psib_r +  h^{(1)}_r  X_r\nonumber\\
H_1 &= \sum_r h^{(0)}_r i  \psib_r \psi_{r+1} +  h^{(1)}_r X_r\nonumber\\
H_2 &= \sum_r h^{(0)}_r i   \psib_r \psi_{r+1}X_r+  h^{(1)}_r  i\psi_r \psib_{r+1}X_r Z_{r-1} Z_{r+1}  \nonumber\\
H_3 &=\sum_r h^{(0)}_r i  \psi_r \psib_{r+1} +  h^{(1)}_r X_r Z_{r-1} Z_{r+1}\nonumber\\
H_4 &= \sum_r h^{(0)}_r i   \psi_r \psib_r+  h^{(1)}_r X_rZ_{r-1} Z_{r+1}\nonumber\\
H_5 &= \sum_r h^{(0)}_r i  \psib_r \psi_{r+1} +  h^{(1)}_r  X_r Z_{r-1} Z_{r+1}\nonumber\\
H_6 &= \sum_r h^{(0)}_r i   \psi_r \psib_{r+1}X_r+  h^{(1)}_r  i\psib_r \psi_{r+1}X_r Z_{r-1} Z_{r+1}  \nonumber\\
H_7 &= \sum_r h^{(0)}_r i  \psi_r \psib_{r+1} +  h^{(1)}_r X_r.\nonumber\\
\end{align}
These are all commuting stabilizer Hamiltonians, each l-bit taking values $\pm 1$, and the stabilizers commute with both fermion parity symmetry and time reversal, so they obey the conditions discussed in \secref{ss:antiunitary}. To get any of the $\mathbb{Z}_8\times\mathbb{Z}_2\times\mathbb{Z}_2$ worth of Floquet phases it pays to consider three auxiliary Hamiltonians

\begin{align}
H^{A}&= \sum_r   Z_r Z_{r+1}  \nonumber\\
H^{B}&= \sum_r   i\psi_r\psib_r \nonumber\\
H^{C}&= \sum_r  i\psi_r\psib_{r+1}, \nonumber\\
\end{align}
and drives of the form 

\begin{align}
U_f  &= e^{- i t_0 K_0/2}  e^{- i t_1 K_1/2}  e^{- i t_2 K_2/2}     e^{- i t_3 K_3}   \nonumber\\ 	
	&\times e^{- i t_2 K_2/2}   e^{- i t_1 K_1/2}   e^{- i t_0 K_0/2}
\end{align}
To obtain a drive $(x,y,z)\in \mathbb{Z}_8 \times \mathbb{Z}_2 \times \mathbb{Z}_2$ first pick $K_3=H_x$, the Hamiltonian with bulk order `$x$'. For such a fixed choice of $x$, the various choices of $K_i$ for the four possible $(y,z)$ are summarized in \tabref{tab:BDIdrives}, as are the corresponding forms of the $v_L$. In the language of \secref{ss:antiunitary}, the four possible $(y,z)$ correspond to the four possible twisted representations, with $(y,z)=(\kappa(P),\kappa(\mathscr{T}))$.  As before, for all the Hamiltonians involved in the above drives, we will ensure $h_r$ is say log-normal distributed with mean $1$. All of the drives so constructed are `fixed-point' in the sense that they have a complete set of exactly local integrals of the motion in the bulk.

\section{Concluding remarks}\label{s:conclusion}
We have put forward a classification scheme for many-body localized Floquet SPT states in one spatial dimension with finite unitary on-site symmetries. In our scheme Floquet drives are classified by $\text{Cl}_G\times\mathcal{A}_G$ where Cl$_G$ is the non-driven SPT classification, and $\mathcal{A}_G$ is a set of 1d representations of $G$ (i.e., $H^{1}(G,\text{U}(1))$. We have also tentatively extended these methods to cases with $\mathscr{T}^2=1$ time reversal for which  $G=G'\times \mathbb{Z}^{\mathscr{T}}_2$ where $G'$ is unitary. In these cases the classification is again of form $\text{Cl}_G\times\mathcal{A}_G$, but $\mathcal{A}_G =H^{1}(G',\text{U}(1))\times \mathbb{Z}_2$. In addition, we have given examples of idealized drives which realize the predicted putative Floquet phases. 

The current work can be extended in several directions. There is the possibility of investigating driven analogues of disordered anyon chains\cite{Vasseur15}. In a sequel to this work \cite{vonKeyserlingkSondhi16b}, we use a similar toolkit to classify the possible symmetry broken Floquet phases 
in 1d; due to localization such order can indeed be observed in apparent violation of the standard theorems on broken symmetry and
dimensionality. The extension of these results to higher dimensions is a fit subject for study,  especially given recent questions over the existence of MBL phases in $d>1$ . Also left to future work is the detailed connection
between the edge-based classification used in this paper and the bulk diagnostics used in Ref.~\onlinecite{Khemani15}. Finally there is
the challenge of understanding the dynamical stability of these new phases for realistic drives en route to proposals for realizing
and detecting them in experiments.

\textit{Note:} Three closely related independent works appeared shortly after we posted this manuscript\cite{Else16,Potter16,Harper16}. The first two of these references phrase the classification in terms of the second cohomology $H^2(G \rtimes \mathbb{Z},\text{U}(1))$ where $G$ is the global on-site symmetry group and $\mathbb{Z}$ is to be identified with time translation by one Floquet period. This interpretation is similar to that in our discussion in \secref{s:algebraic} -- indeed, the operators $v_{L,R}$ can be thought of as time translations local to the $L,R$ edges respectively.  With this interpretation in mind,  \secref{s:algebraic} establishes how time-translation acts together with the other symmetries at the edge of the system. This is in fact the same thing as calculating the projective representations of the total symmetry group $G \rtimes \mathbb{Z}$ -- that is, calculating $H^2(G \rtimes \mathbb{Z},\text{U}(1))$.

\acknowledgements
We thank V. Khemani, R. Moessner and A. Lazarides for many discussions and for collaboration (with SLS) on prior work. We would
also like to thank R. Roy for generously sharing his unpublished work on the topological classification of free fermion drives. We thank A.~Potter for alerting us to his work\cite{Potter16}. CVK is supported by the Princeton Center for Theoretical Science. SLS would like to acknowledge support from the NSF-DMR via Grant No.~1311781 and the Alexander von Humboldt Foundation for support during a stay at MPI-PKS where this work was begun.

\begin{appendix}

\section{Locality of $f$}\label{app:flocal}
 Consider a Floquet unitary $U_f = e^{- i f}$ on a system without boundary, with a complete set of local integrals of motion.  Then $f$ may be written as a function of these local conserved quantities $f=f(\{\Gamma_r\})$. We now give a sketch of an argument that $f$ may be chosen to be a  local function of these conserved quantities.  We use the fact that $U_f$ is \textit{a priori} the result of a local unitary evolution. Local unitary evolutions may be well approximated by finite depth quantum circuits\cite{Chen10} -- for simplicity let us assume that $U_f$ is a binary quantum circuit of depth $d$ (e.g., \figref{circuit}), where $d$ remains finite in the thermodynamic limit.

As a warmup, we argue that for local conserved quantities $\Gamma_x, \Gamma_y$ separated in excess of distance $d$, and well in excess of the size of the local conserved quantities $\sim \xi$,  the Floquet unitary can be written as $U_f = A B$ where $A$ depends on $\Gamma_x$ but not $\Gamma_y$, and  $B$ depends on $\Gamma_y$ but not $\Gamma_x$. 

Consider a set of local conserved quantities associated with site $x$. These quantities commute with one another, so we can find a simultaneous eigenbasis for the local Hilbert space. Label the distinct possible lists of simultaneous eigenvalues by integers $\lambda \in \{1,2,\ldots, l \}$.  There may be degeneracies, so the eigenvectors corresponding to $\lambda$ are of form $E_{\lambda} = \{ v_{\lambda,\alpha_1},\ldots , v_{\lambda, \deg_\lambda}\}$. Let $S^\sigma_x$ be a unitary which permutes all of the eigenvectors $\cup_\lambda E_\lambda $ according to some permutation cycle $\sigma$.  As $\sigma$ permutes eigenvectors, it may also permute local eigenvalues through an action we denote $\lambda \rightarrow \sigma(\lambda)$.

$S^\sigma_x$ is a local operator because it only permutes some eigenvectors in the local Hilbert space. As $S^\sigma_x$  is local to $x$ and $U_f$ is low depth, the commutator $[U_f:S^\sigma_x]$ is local to $x$ (this is a Lieb-Robinson type bound -- except in the quantum circuit formalism there is no exponentially decaying tail). If $S^\sigma_x$  was $\xi$ exponentially localized before, then the commutator is safely localized around $x$ certainly when considering distances much larger than $\xi$ and at least $d$! We call this the `smearing' length scale $\zeta$. 

If $|x-y|>\zeta$, then operators localized around $y$ should commute with $[U_f:S^\sigma_x]$ i.e., 
\be\label{eq:doublecomm}
[[U_f:S^\sigma_x]:S^\tau_y] \sim 1
\ee
where $\tau$ is any permutation. Consider components of this equation in the eigenbasis of local conserved quantities. $U_f$ depends on all the local conserved quantities in general, but we concentrate on the dependence of those conserved quantities near $x,y$, writing $U_f = U_f (\lambda_x, \lambda_y)$ where $\lambda_{x,y}\in \{1,\ldots, l \}$, and suppressing the other labels for now. The equation \eqnref{eq:doublecomm} reads

$$
\frac{U_f (\lambda_x, \lambda_y) U_f (\sigma(\lambda_x), \tau(\lambda_y))}{U_f (\sigma(\lambda_x), \lambda_y) U_f (\lambda_x, \tau(\lambda_y) )} =1.
$$
Pick $\sigma,\tau$ such that $\sigma(\lambda_x)=1$ and   $\tau(\lambda_y)=1$ , and rearrange to find

\be \label{eq:Ufanz}
U_f (\lambda_x, \lambda_y) = U_f (\lambda_x, 1 ) \times\frac{U_f (1, \lambda_y) }{ U_f (1,1)}  \punc{.}
\ee
The first factor on the right hand side depends on $\lambda_x$ but not $\lambda_y$ while the second depends on  $\lambda_y$ but not $\lambda_x$ as required.

To argue that $f$ can be chosen to be local,  concentrate on the factor $U_f (\lambda_x, 1 )$. Now this implicitly depends on the values of other conserved quantities. Using the same reasoning as above and our previous result \eqnref{eq:Ufanz}, we can factorize out the dependence of any $\lambda_z$ for $|z-x|>\zeta$, namely (again suppressing dependence on other conserved quantities)

$$
U_f (\lambda_x, \lambda_y, \lambda_z) = {U_f (\lambda_x, 1,1 )} \times \frac{U_f (1, \lambda_y,\lambda_z) }{ U_f (1,1,1)}  \punc{.} $$
where the first term  does not depend on $\lambda_{y,z}$ , while the last two terms do not depend on $\lambda_x$. We can proceed inductively to show that
$$
U_f  = \underbrace{U(\lambda_x, \lambda_{\text{far} } =1  )}_{A_x}  \underbrace{\frac{U(\lambda_x = 1, \lambda_{\text{far}} )}{U(\lambda_x = 1, \lambda_{\text{far}} =1 ) }}_{B}
$$ 
where $\lambda_{\text{far}}$ are those labels at sites further than $\zeta$ from $x$, other labels are kept implicit, $A_x$ depends on $\lambda_x$ and only conserved quantities within $\zeta$ of $x$, while $B$ does not depend on $\lambda_x$. In particular, $U_f  A^{-1}_x (=B) $ is a local unitary which does not depend on $\lambda_x$. Moreover, one can verify for any site $x'$ that $[B:S^{\sigma}_{x'}]$ is localized to within the same length scale $\zeta$ of $x'$ -- with some thought,  this follows from the fact that $U_f$ has smearing scale $\zeta$, and that $B$ is a product of $U_f$ factors with different combinations of conserved quantities set to $1$. With this new site $x' \neq x$, repeat the procedure above (but using $U_f  A^{-1}_x$ instead of $U_f$)  to get

$$
U_f  A^{-1}_x = A_{x'}  B' 
$$
where $A_{x'}$ depends only on conserved quantities near to $x'$, and $B'$ is a local unitary which does not depend on $\lambda_{x},\lambda_{x'}$ and which also has smearing scale bounded by $\zeta$. Once can repeat this process inductively, cycling through all sites until eventually we find
$$
U_f   \propto \prod_{x} A_{x}
$$
where  $A_x$ is a U$(1)$ function of conserved quantities within $\zeta$ of $x$, so can be written $e^{-i f_x}$ where $f_x$ is a real function of conserved quantities within $\zeta$ of $x$. As all the conserved quantities commute we have

$$
U_f  = e^{-i \sum_x  f_x}
$$
up to a global U$(1)$ phase. We see that $f$ so defined here is at most a `$k$-local' Hamiltonian where $k = \zeta$. The argument above is clearly non-rigorous. We require a more careful analysis of the importance of exponentially small corrections in the iterative procedure outlined above.

\section{Form of $U_f$}\label{app:misclemmas}
In this section, we first provide supporting arguments for {\bf (i)} in the main text in \secref{ss:provei} -- our arguments will rely heavily on the results of \appref{app:flocal}. Then in \secref{eq:restrictionlemmas} we consider starting with a Floquet drive on a closed system, and show that upon restricting the Floquet drive to a subsystem, the resulting unitary also takes the canonical form \eqnref{eq:decomprough}. 

\subsection{Floquet MBL unitaries on systems with boundary}\label{ss:provei}
\begin{figure}[h]
\includegraphics[width=0.99\columnwidth]{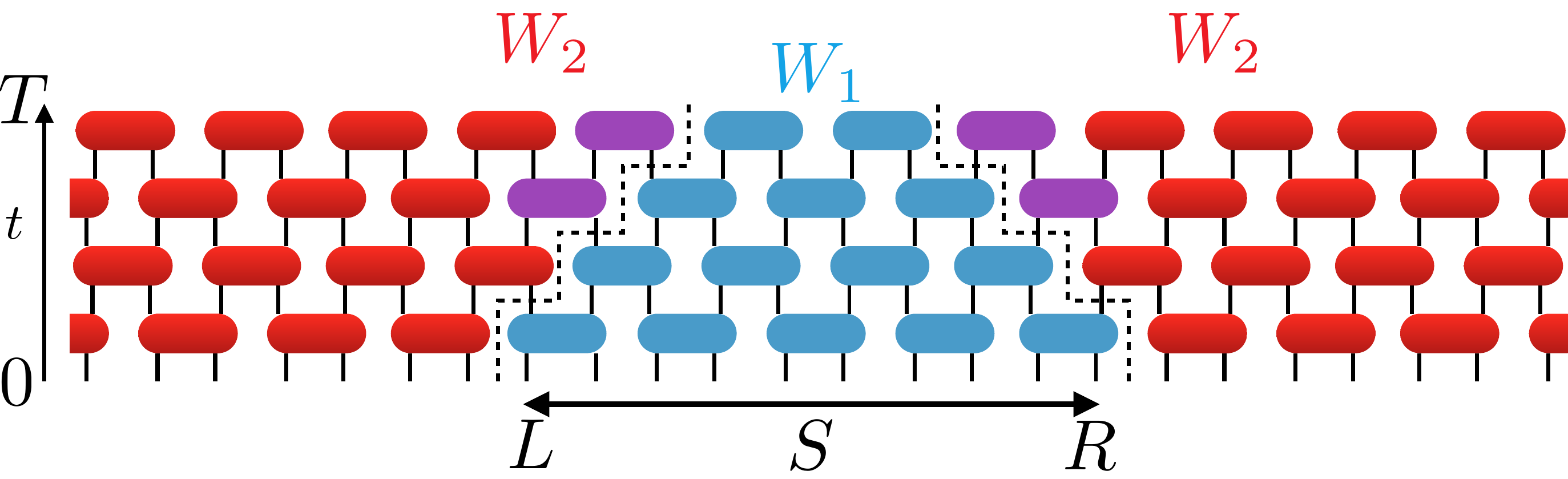}
\caption{(Color Online): Shows a circuit diagram for local unitary $U(T)$. We identify parts of the circuit $W_{1}$ (blue) and $W_2$ (purple and red).   }
\label{circuit}
\end{figure}
The main goal of this section is to support the claim {\bf(i)} in the main text, that Floquet unitaries with a complete bulk MBL order, can be put into the form \eqnref{eq:decomprough}. We assume that $U_f$ is a local unitary on a system $[L,R]$, with a complete set of l-bits in the bulk. Deep in the bulk, we know that $U_f$ is depends only on the bulk l-bits. Moreover using the results of \appref{app:flocal}
$$
U_f = \mathcal{O}_{L,R} e^{- i h}
$$
where $h$ is a local functional of the bulk l-bits, and $\mathcal{O}_{L,R}$ is some unitary acting near the boundary $L,R$ of the system. However, as $U_f$ is a local unitary circuit, and $L,R$ are very distant from one another, it follows that  $\mathcal{O}_{L,R}$ must factorize as $\mathcal{O}_{L} \mathcal{O}_{R} $
$$
U_f = \mathcal{O}_{L} \mathcal{O}_{R}  e^{- i h}
$$
where $\mathcal{O}_{L,R}$ are local to the left/right part of the system respectively. Consider a region $S'=[L+a, R-a]$ where $a$ is much larger than the support of $\mathcal{O}_{L,R}$ as well as the typical size $\xi$ of the conserved quantities. We can write
\be\label{eq:redefh}
h = h_{S'} + g_L + g_R
\ee
where  $h_{S'}$ are all those terms in $h$ which involve only conserved quantities in $S'$. While by locality $g_{L,R}$ involve only conserved quantities on the LHS/RHS of the system respectively (up to the exponentially small corrections mentioned in \appref{app:flocal}). Note that all three terms on the RHS of \eqnref{eq:redefh} commute with one another because they only involve the conserved quantities. Note too that $\mathcal{O}_{L}, \mathcal{O}_{R}$ commute with $h_{S'}$ because as operators they have disjoint support (and $h$ always fermion parity even). From the discussion of $g_L,g_R$ above, $v_{L,R} \equiv \mathcal{O}_{L,R} e^{- i g_{L,R}} $ will commute with $f \equiv h_{S'}$  therefore 
 $$
 U_f=v_L v_R e^{-i f(\Gamma)}
 $$
where $v_L,v_R$ commute with $f$, and indeed all conserved quantities with support in the `bulk' $S'$, as required. Having provided arguments supporting {\bf (i)}, we now give a method for deciding whether or not a Floquet drive defined on a system without boundary is in a trivial or non-trivial class.

\subsection{Characterizing Floquet MBL unitaries on systems without boundary}\label{eq:restrictionlemmas}
We start with a definition. 
\begin{definition}\label{def:restriction}
Given a many-body unitary evolution $U(t)=\mathcal{T}(e^{-i \int^{t}_{0} dt' H(t') } )$ with $H(t')$ a family of local bounded Hamiltonians on a closed system, we define the restricted unitary $U_{S}(t)\equiv\mathcal{T} (e^{-i \int^{t}_{0} dt' H_{S}(t') } )$ where $H_{S}(t')$ are those terms in the Hamiltonian acting exclusively on subsystem $S$.
\end{definition}

We will show that if one takes a unitary circuit with full bulk MBL order on a manifold without boundary, and restrict to a system with boundary, the resulting unitary $U_S$, can be put into the desired form \eqnref{eq:decomprough}. We will make use of the arguments in the previous section which showed that on a closed system, $U_f = e^{- i f}$ where $f$ is a local functional of bulk conserved quantities. First, a useful technical lemma.

\begin{lemma}\label{lem:vLvRlemma}
Consider a local unitary circuit $W(t)$ of depth $d$. If $W(T)=1$ on a closed system, then restricting to subsystem $S=[L,R]$ we find $W_S(T) = v_L v_R$ where $v_L,v_R$ are unitaries localized within $d$ of $L,R$. 
\end{lemma}
\begin{proof}
Consider those circuit elements in the future Cauchy development of $[L,R]$ (the blue region in \figref{circuit}). Denote the unitary formed by multiplying out these circuit elements by $W_{1}$ \footnote{In the continuous time language, this is approximately the same as $\mathcal{T} e^{-i \int^T_0 dt H_{S(t)}}$ where $S(t)=[L+ct, R-ct]$, the notation $H_{S(t)}$ denotes those terms in the Hamiltonian involving only sites in region $S(t)$, and $c$ is the Lieb-Robinson velocity.}. Denote the rest of the unitary circuit by $W_2$ (red and purple in \figref{circuit}) . Then $W(T)=W_2 W_1 =1$, and notably $W_1 = W^{-1}_2$. However $W_1$ has support in $S=[L,R]$ while $W^{\pm 1}_2$ has support on a different set, namely the complement of $[L+d,R-d]$. The only possible resolution is that  both $W_1$ and $W_2$ have support only in the intersection of these two sets, namely $[L,L+d)\cup (R-d,R]$. By \defref{def:restriction} we have $W_S (T) = W_3 W_1$ where $W_3$ is formed of those circuit elements with support on $[L,R]$ but not in $W_1$  (purple circuit elements in \figref{circuit}). Thus $W_3$ has spatial support in $[L,R]$ within $d$ of $L$ or $R$. The same statement is true of $W_1$ and hence also true for $W_S(T)=W_3 W_1$. Hence $W_S (T)= v_L v_R$ where $v_L$ has support in $[L,L+d] $ and $v_r$ has support in $ [R,R-d]$. 

\end{proof}

\begin{lemma}\label{BulkEigenstates}
A local Floquet unitary $U(T)$ with eigenstate order restricted to subsystem $S$ takes form $U_S (T)=v_L v_R e^{-i f'(\Gamma)}$ where $v_L,v_R$ are unitaries localized near the boundary, and  $[v_{L/R}:f]=  1$.
\end{lemma}
\begin{proof}
As $U(T)$ has eigenstate order, and is low-depth, we assume we can write it as a local functional of local conserved quantities $e^{-i f( \Gamma)}$. Let $S$ be an extensive subregion of the system. The unitary circuit $W(t)$ formed by concatenating $U(t)$ and a circuit corresponding to $U'(t)=e^{i t f }$. Now the unitary circuit $W(T) = 1$, and is local by construction. Hence by \lemref{lem:vLvRlemma} it has the property $W_S(T)=v_L v_R$. On the other hand, from \defref{def:restriction} we have $W_S(T)= U_S(T) e^{ i f_S}$, where $f_S$ is just  $f$ restricted to those terms involving only conserved quantities in $S$. Hence we find $U_S(T) = v_L v_R e^{ -i f_S}$. At this stage it is not clear that $v_L,v_R$ commute with $f_S$. To make this clear, consider a region $S'=[L+a, R-a]$ where $a$ is much larger than the depth of the circuit, and the size $\xi$ of the conserved quantities. We can write
$$
f_{S} = f_{S'} + g_L + g_R
$$
 where $g_L$ involves only conserved quantities on the LHS of the system, while $g_R$ involves those on the right and all three terms on the RHS commute with one another because they only involve the conserved quantities. Note too that $v_L,v_R$ commute with $f_{S'}$ because as operators they have disjoint support (and $f$ always fermion parity even). From the discussion of $g_L,g_R$ above, $v_L,v_R$ continue to commute with $f_{S'}$  if we redefine $v_L \rightarrow v_L e^{- i g_L}$ and  $v_R \rightarrow v_R e^{- i g_R}$, in which case 
 $$
 U_S=v_L v_R e^{-i f_{S'}(\Gamma)}
 $$
where $v_L,v_R$ commute with $f_{S'}$ as required.

\end{proof}

\section{Alternative characterization of pumped charge for unitary symmetry groups $G$}\label{app:pumpedcharge}

Here we give a slightly different definition of pumped charge which does not require the assumption that we can decompose $U_S=v_L v_R e^{- i f}$. In the following we merely assume that $U_S$ is a local unitary with exact eigenstate order, and that the SPT order is many-body localizable and characterized by a string order parameter.

\begin{definition}
A string order operator is a unitary function of $G$ of form $\Gamma^{(c)}_{l,r}(g)=\mathcal{O}^{(c)}_l  \mathcal{O}^{(c)}_r \prod_{s\in (l,r)} V_s (g)$ where $\mathcal{O}^{(c)}_{l,r}$ are unitary operators localized (with some correlation length $\xi$) near $l,r$ respectively, and $V_s(g)$ is the on-site unitary symmetry operator. 
\end{definition}

\begin{definition}
We say unitary $U$ has (exact) eigenstate order $c\in \text{Cl}_G$ if it has a complete set of local conserved quantities taking the form $\Gamma^{(c)}_{l,r}(g)$ with $g \in G$, and $l,r$ arbitrary sites in the system.  Equivalently, $U=e^{- i f(\{\Gamma^{(c)}\})}$ where $f$ is a functional of a complete subset of all the local conserved quantities. 
\end{definition}

\begin{definition}
 Given subsystem $S=[L,R]$ we define special string order operator $\Gamma^{(c)}_L(g)= \mathcal{O}^{(c)}_M \prod_{r\in{[L,M)}}V(g)$ where $M$ is some point extensively far in the bulk (e.g., half-way along $S$). Similarly  $\Gamma^{(c)}_R(g)= \mathcal{O}^{(c)}_M \prod_{r\in{(M,R]}}V(g)$. 
\end{definition}

\begin{lemma}\label{lem:GLcommuteswithbulk}
$\Gamma^{(c)}_L(g)$ commutes with all local conserved quantities $\Gamma$ entirely in $S'$.
\end{lemma}
\begin{proof}
$\Gamma^{(c)}_L(g)$ is defined by writing down a conserved quantity on the original uncut system $\Gamma^{(c)}_{x,M}(g) $ where $x$ is many $\xi$ to the left of $S$, and restricting this unitary to $[L,R]$. Indeed $\Gamma^{(c)}_L(g) = O \Gamma^{(c)}_{x,M}(g)$ where $O$ is completely outside of $S$.  On the original uncut system, $\Gamma^{(c)}_{x,L}(g)$ will \text{a priori} commute with all the conserved quantities $\Gamma$ in $S'$.  Note then that $[\Gamma^{(c)}_L(g): \Gamma]= [ O \Gamma^{(c)}_{[x,M]}(g): \Gamma]=[ O: \Gamma]$. As the complement of $S$ is many correlation lengths away from $S'$ (and $\Gamma$ necessarily fermion parity even), we must have  $[ O: \Gamma]=1$. Hence $[\Gamma^{(c)}_L(g): \Gamma]=1$. 
\end{proof}

\begin{lemma}\label{lem:uptophase}
$\Gamma^{(c)}_L(g)$ commutes with $U_S$ up to a phase. Moreover, using \textbf{(A)}, this phase is equal to the older definition of pumped charge $\kappa(g)=[\Gamma^{(c)}_{L}:U_S]=[\Gamma^{(c)}_{R}:U_S]^{-1}$ if we assume $U_S=v_L v_R e^{-i f}$
\end{lemma}
\begin{proof}
$\Gamma^{(c)}_L(g)$ is supported almost entirely on some interval $[L,M]$. Pick a point $x$ in this interval but many $\xi$ from $L,M$. Then to good approximation $\Gamma^{(c)}_L(g) = V_{L,x}(g) q_x \Gamma^{(c)}_{x,M}(g) $ where $\Gamma^{(c)}_{x,M}(g)$ commutes with $U_S$ and $q_x$ local to $x$ . As a result $[\Gamma^{(c)}_{L}:U_S] =  [V_{L,x}(g) q_x:U_S] = \mathfrak{g}(x)$. It follows from the fact $U_S$ is low depth and symmetric that $\mathfrak{g}(x)$ is some operator with support near $x$ . Therefore $[\Gamma^{(c)}_{L}:U_S] = \mathfrak{g}(x) $ for any $x \in [L,M]$ but many $\xi$ away from the end-points. As the LHS does not depend on $x$ this implies $\mathfrak{g}(x)$ is a pure phase. Now as $\Gamma^{(c)}_{L}$ commutes with all of the conserved quantities well in the bulk we have $[\Gamma^{(c)}_{L}:U_S] = [\Gamma^{(c)}_{L}:v_L v_R] = \eta  [\Gamma^{(c)}_{L}:v_L] $ where $\eta = 1$ unless both $\Gamma^{(c)}_{L},v_R$ are fermion parity odd in which case $\eta = -1$. Now from the form of the string operator we know that $[\Gamma^{(c)}_{L}:v_L] = \eta [V(g)_S:v_L]  = \eta \kappa_L(g)$ where $\eta$ is the same as above because $v_L ,v_R$ have the same fermion parity. Hence $[\Gamma^{(c)}_{L}:U_S] = \kappa(g)$.
\end{proof}

\begin{lemma}
The pumped charge $\kappa(g)$ defined above is robust under $U_S \rightarrow U_S w_L w_R$ where $w_{L,R}$ are local (compared to the system size) symmetric unitaries which act only near the $L,R$ end of the system respectively.
\end{lemma}
\begin{proof}
This follows  from the previous lemma.  On sites $s$ in the support of $w_L$,  $\Gamma^{(c)}_{L}$ acts  like $V_{s}(g)$. But $w_L$ is symmetric (and assumed parity even in a system with fermions \footnote{In fermionic SPTs, fermion parity is always a symmetry, so as $W_{L,R}$ are symmetric they must also be fermion parity even\cite{Gu14}.}). Hence $\Gamma^{(c)}_{L}$ commutes with $w_L$. Therefore $[\Gamma^{(c)}_{L}:U_S] = \kappa(g)$ remains unchanged. 
\end{proof}

%
%

\begin{lemma}\label{trivialdrives} Time independent Hamiltonian (TIH) drives have trivial pumped charge. 
\end{lemma}
\begin{proof}
For a time independent Hamiltonian drive $U_{S}(T)=e^{-i H_{S}T}$ where $H_S$ is a symmetric local Hamiltonian. This Hamiltonian is assumed to have an eigenstate order $c$. It follows that from the existence of such a string order that $[\Gamma^{(c)}_{L}(g):H_S]=0$. As a result $U_S$ has the same eigenstate order, and $[\Gamma^{(c)}_{L}(g):U_S]=1$. Hence $\kappa(g)=1$. 
\end{proof}
\section{Edge structure for Class D}\label{s:classDedge}
In this section we use a different method to show how the Floquet phases in \secref{s:motivatingexample} arise. Suppose we are given a fermion parity symmetric, local unitary Floquet circuit $U$ on a large closed system.  Suppose further that the Floquet unitary $U(T)$ has eigenstate order -- that is to say, there is a complete set of conserved quantites of the (approximate) form 
$$
\Gamma_{l,r}^{(c)}(g=-1)=\begin{cases}
\prod_{l<s<r}i\psi_{s} \psib_{s} & c=1\\
i\psi_{l} \ds{\prod_{l<s<r}i\psi_{s} \psib_{s}} \psib_{r}& c=-1
\end{cases}
 $$
where $c=\pm1$ corresponds to trivial/topological eigenstate order respectively. On a system with boundary, we can arrange things so that
\be
\mathcal{A}=\text{gen}\left\{\underbrace{v_{L} v_{R} e^{-i f}}_{U_{S}},v_L,v_R, \Gamma^{(c)}(g)_{L},   \Gamma^{(c)}(g)_{R}, \Gamma^{(c)}_{b}(g)\right\}
\ee
are a complete set of operators. The notation $\Gamma^{(c)}(g)_{L,R}$ is explained in \appref{app:pumpedcharge}.  Our goal is to find the dimension and $U_{S}$ eigenvalues of a minimal representation of this algebra. Formally, we are looking for $\mathcal{A}/Z(\mathcal{A})$ where $Z$ is the center of the algebra. But in the present case, due to the bulk eigenstate order, $Z(\mathcal{A})$ contains all the bulk string operators, so we are indeed concerned only with $\mathfrak{a}/Z(\mfa)$ where

\be\label{eq:classDalg}
\mathfrak{a}=\text{gen}\left\{v_{L},v_{R}, \Gamma^{(c)}(g)_{L},\Gamma^{(c)}(g)_{R} \}\right\}
\ee

The commutation relations of this algebra are
\begin{align}
[\Gamma^{(c)}(-1)_{L}: U_S]= \kappa(g) \n
[\Gamma^{(c)}(-1)_{L}: \Gamma^{(c)}(-1)_{R}]= c
\end{align}
where $\eta \in \mathcal{A}_G = \mathbb{Z}_2$ and $c\in 
\mathbb{Z}_2$. One now asks what are the minimal dimension representations of this algebra? The answer is $d =\frac{1}{4}(9+c-3p - 3 c p)$ where $p=\kappa(-1)$. 
\subsection{Class D edges}
Here we work out the representation theory of $\mathfrak{a}$. The commutation relations for this algebra are
\begin{align*}
[\Gamma_{L}^{(c)}:\Gamma_{R}^{(c)}] & =c & [\Gamma_{L}^{(c)}:v_{L}v_{R}] & =p\\
[v_{L}:v_{R}] & =p & [\Gamma_{L}^{(c)}:v_{L}] & =(-1)^{\delta_{c=p=-1}}p\\
\end{align*}

\paragraph*{$c=1,p=1$:} In this cases, all the generators of $\mfa$ commute. Therefore $\mfa/Z(\mfa)=\{1\}$ is trivial, and there are no protected degeneracies in the Floquet spectrum.  Therefore $\mfa/Z(\mfa)=\{1\}$ is trivial, and there are no protected degeneracies in the Floquet spectrum.

\paragraph*{$c=1,p=-1$:} In this cases, none of the generators of $\mfa$ commute. To elucidate the edge structure, we find a maximal commuting sub-algebra of $\mfa$. It is a convenient to choose $\mathfrak{b}=\langle v_Lv_R,\Gamma_L \Gamma_R\rangle$ because the eigenvalues of $v_L v_R$ are up to a phase (from the bulk) just the eigenvalues of $U_S$. Starting with a simultaneous eigenstate $\mid \psi \rangle$ for $\mathfrak{b}$, we get a minimal representation of size $4$ summarized here:

\begin{tabular}{|c|c|c|}
\hline 
 & ${v_{L}v_{R}}$ & $\Gamma_{L}\Gamma_{R}$\tabularnewline
\hline 
$\mid\psi\rangle $ & $u$ & $u'$\tabularnewline
\hline 
$\Gamma_{L}$$\mid\psi\rangle $ & $-u$ & $u'$\tabularnewline
\hline 
$V_{L}$$\mid\psi\rangle $ & $u$ & $-u'$\tabularnewline
\hline 
$\Gamma_{L}V_{L}$$\mid\psi\rangle $ & $-u$ & $-u'$\tabularnewline
\hline 
\end{tabular}

\paragraph*{$c=-1,p=1$:} In this case  $Z(\mfa)=\langle v_L,v_R\rangle$, so the remaining sub-algebra is $\mfa/Z(\mfa)=\langle \Gamma_L,\Gamma_R\rangle$. To elucidate the edge structure, we find a maximal commuting sub-algebra of $\mfa/Z(\mfa)$. It is a convenient to choose $\mathfrak{b}=\langle \Gamma_L \Gamma_R \rangle$. Starting with a simultaneous eigenstate $\mid \psi \rangle$ for $\mathfrak{b}$, we get a minimal representation of size $2$.
\begin{tabular}{|c|c|c|}
\hline 
 & ${v_{L}v_{R}}$ & $\Gamma_{L}\Gamma_{R}$\tabularnewline
\hline 
$\mid\psi\rangle $ & $u$ & $u'$\tabularnewline
\hline 
$\Gamma_{L}$$\mid\psi\rangle $ & $u$ & $-u'$\tabularnewline
\hline 
\end{tabular}
\paragraph*{$c=-1,p=-1$:} In this case  $Z(\mfa)=\langle \Gamma_Lv_L,\Gamma_R v_R\rangle$, so the remaining sub-algebra is $\mfa/Z(\mfa)=\langle V_L,V_R\rangle$. A maximal commuting sub-algebra is $\mathfrak{b}=\langle v_L v_R \rangle$. Starting with a simultaneous eigenstate $\mid \psi \rangle$ for $\mathfrak{b}$, we again get a minimal representation of size $2$.
\begin{tabular}{|c|c|}
\hline 
 & $v_{L}v_{R}$\tabularnewline
\hline 
$\mid\psi\rangle $ & $u$\tabularnewline
\hline 
$v_{L}$$\mid\psi\rangle $ & $-u$\tabularnewline
\hline 
\end{tabular}

\subsection{MBL Binary drives realizing the Floquet phases}
The four possible Floquet phases here described can be realized using binary drives, as demonstrated in \cite{Khemani15}. A binary drive involving Hamiltonians $H_1,H_2$ and times $t_1,t_2$ is a unitary matrix function of time $U(t)$

\[
U(t)\equiv\begin{cases}
e^{-iH_{1}t} & 0\leq t<t_{1}\\
e^{-iH_{2}(t-t_{1})}e^{-iH_{1}t_{1}} & t_{1}\leq t<t_{1}+t_{2}
\end{cases}
\]
In the context of class D, we set 

\begin{align*}
H_{1}= & -\sum ih_{i}\psi_{i}\bar{\psi}_{i}\\
H_{2}= & -\sum iJ_{i}\psi_{i}\bar{\psi}_{i+i}
\end{align*}
Setting $\overline{J_i}=\overline{h_i}=1$, we can get the full $\mathbb{Z}_2\times \mathbb{Z}_2$ classification by using $(t_1,t_2)=(\frac{\pi}{4},0),(0,\frac{\pi}{4}),(\frac{\pi}{4},\frac{\pi}{2}),(\frac{\pi}{2},\frac{\pi}{4})$ which correspond respectively to obtain $(c,p)=(1,1),(-1,1),(1,-1),(-1,-1)$. 

\section{Representation theory of $\mathbb{Z}_2\times \mathbb{Z}_2$ edge in general, using string order method}\label{app:Z2Z2full}
Continuing from \secref{s:Z2xZ2}, let us discuss the string order in $\mathbb{Z}_2\times \mathbb{Z}_2$ PM and SPT states.  The different forms of l-bits lead to a different kinds of string order in the two resulting SPT phases. Multiplying the l-bits together, notice that the eigenstates of the paramagnet can be chosen to be eigenstates of the string operators
\be\label{eq:PM}
\prod^r_{s=l} X^{1}_s,\,\,\,\, \prod^r_{s=l} X^{2}_s\punc{,}
\ee
The two operators can be though of as local generators of the $\mathbb{Z}_2\times \mathbb{Z}_2$ symmetry, corresponding to group elements $(-1,1),(1,-1)$ respectively. On the other hand, the eigenstates of the SPT can be chosen to be eigenstates of
\be\label{eq:SPT}
Z^{2}_{l-1} \prod^r_{s=l} X^{1}_s Z^{2}_{r+1},\,\,\,\,Z^{1}_{l} \prod^r_{s=l} X^{2}_s Z^{1}_{r-1}\punc{.}
\ee
These operators (away from $l,r$) also look like $(-1,1),(1,-1)$ symmetry generators. Now move away from from the fixed point and consider a disordered scenario with the same symmetry group, and a complete set of l-bits.  Away from the fixed points, we still expect there to be string operators which commute with the Hamiltonian, which are unitary functions of $G$ of form $\Gamma_{l,r}(g)=\mathcal{O}_l   \prod_{s\in (l,r)} V_s (g) \mathcal{O}_r$ where $\mathcal{O}_{l},\mathcal{O}_{r}$ are unitary operators localized (with some characteristic length scale $\xi$) near $l,r$ respectively, and $V_s(g)$ is the on-site unitary symmetry operator -- corresponding to $X^{1,2}_s$ for $g=(-1,1),(1,-1)$ respectively.  We will say that the Hamiltonian is in the trivial MBL PM phase if it commutes with a family of string operators which at large distances take a form approaching \eqnref{eq:PM} while it is said to be in the SPT phase if they take the form \eqnref{eq:SPT}. With these string operators in mind, we remind the reader how to define an action of the symmetry group at each edge (see \appref{app:pumpedcharge}). 

\begin{definition}
 Given subsystem $S=[L,R]$ we define modify the above string order operators to form $\Gamma_L(g)= \mathcal{O}_M \prod_{r\in{[L,M)}}V(g)$ where $M$ is some point extensively far in the bulk (e.g., half-way along $S$). Similarly  $\Gamma_R(g)= \mathcal{O}_M \prod_{r\in{(M,R]}}V(g)$. 
\end{definition}

In the MBL phase, these operators $\Gamma_L(g),\Gamma_R(g)$ act like symmetry generators at each edge, and can be argued to commute with the bulk conserved quantities (see \lemref{lem:GLcommuteswithbulk}). For this reason, it is natural to identify them with $\widehat{g}_L,\widehat{g}_R$ from the previous section. 

\subsection{Edge structure}
Now suppose we are given a unitary $U_S$ with prescribed bulk eigenstate order `$c$' and corresponding string order operators $\Gamma_{l,r}^{(c)}(g)$. By \lemref{lem:uptophase}, $\Gamma_{R}^{(c)}(g)$ commutes with $U_S$ up to phases characterized by a 1d representation $\kappa(g)$. The set of operators which commute with $U_S$ up to phases are simply
\be
\mathcal{A}=\text{gen}\left\{U_{S}, \Gamma^{(c)}(g)_{L},   \Gamma^{(c)}(g)_{R}, \Gamma^{(c)}(g)_{bulk}\right\}
\ee
Our goal is to find the dimension and $U_{S}$ eigenvalues of a minimal representation of this algebra. Formally, we are looking for $\mathcal{A}/Z(\mathcal{A})$. Due to the bulk eigenstate order, $Z(\mathcal{A})$ is simply generated by the bulk string operators. Hence, we are interested only in the algebra

\be
\mfa=\text{gen}\left\{U_{S}, \Gamma^{(c)}(g)_{L},   \Gamma^{(c)}(g)_{R}\right\}
\ee
The commutation relations of this algebra are 
\begin{align}
[\Gamma^{(c)}(g)_{L}: U_S]&= \kappa(g) \n
[\Gamma^{(c)}(g)_{R}: U_S]&= \kappa^{-1}(g) \n
[\Gamma^{(c)}(g)_{L}: \Gamma^{(c)}(h)_{L}]&= \mu_c(g,h)\n
[\Gamma^{(c)}(g)_{R}: \Gamma^{(c)}(h)_{R}]&= \mu_c^{-1}(g,h)
\end{align}
where $\kappa \in \mathcal{A}_G \sim \mathbb{Z}_2\times \mathbb{Z}_2$ and $c\in \mathbb{Z}_2$. This suggests a Floquet classification $\mathbb{Z}^{\otimes 3}_{2}$.  Here $\mu_c(g,h)=e^{i \pi (g^2 h^1 -g^1 h^2)}= (i_h\omega(g))^{-1}$ if we take  $\omega(g,h)=e^{i \pi (g^1 h^2 -g^2 h^1)}$ (reverting to additive notation). One now asks what are the minimal dimension representations of this algebra? The answer is $d =4$ if either $c=-1$ or $\kappa\neq \text{Id}$. 

Now we work out the representation theory of $\mathfrak{a}$. The commutation relations for this algebra are

\paragraph*{$c=1,\kappa=Id$:} In this cases, all the generators of $\mfa$ commute. Therefore $\mfa/Z(\mfa)=\{1\}$ is trivial, and there are no protected degeneracies in the Floquet spectrum.  Indeed $\mfa/Z(\mfa)=\{1\}$ is trivial.

\paragraph*{$c=1,\kappa\neq Id$:} For a non-trivial character $\kappa$ of $\mathbb{Z}_2 \times \mathbb{Z}_2$, there is a unique $1\neq g_\kappa \in G$ such that $\kappa(g_\kappa)=1$. Let another independent generator by $\bar{g}_\kappa$. With this information, we can show that $Z(\mfa)=\langle \Gamma_L(g_\kappa),\Gamma_R(g_\kappa) \rangle$. Forming $\mfa/Z(\mfa)$, we examine maximal commuting sub-algebra $\mathfrak{b}<\mfa/Z(\mfa)$ generated by representatives $\mathfrak{b}=\langle \Gamma_L(\bar{g}_\kappa)\Gamma_R(\bar{g}_\kappa),v_L v_R\rangle$. See \tabref{tab:PMneq1}.
\begin{table}
\begin{tabular}{|c|c|c|}
\hline 
 & $v_{L}v_{R}$ & $\Gamma_{L}(\bar{g}_{\kappa})\Gamma_{R}(\bar{g}_{\kappa})$\tabularnewline
\hline 
$\mid\psi\rangle $ & $u$ & $u'$\tabularnewline
\hline 
$v_{L}$$\mid\psi\rangle $ & $u$ & $-u'$\tabularnewline
\hline 
$\Gamma_{L}(\bar{g}_{\kappa})$$\mid\psi\rangle $ & $-u$ & $u'$\tabularnewline
\hline 
$\Gamma_{L}(\bar{g}_{\kappa})v_{L}$$\mid\psi\rangle $ & $-u$ & $-u'$\tabularnewline
\hline 
\end{tabular}
\caption{}\label{tab:PMneq1}
\end{table}

\paragraph*{$c=-1,\kappa\neq Id$:} For a non-trivial character $\kappa$ of $\mathbb{Z}_2 \times \mathbb{Z}_2$, there is again a unique $1\neq g_\kappa \in G$ such that $\kappa(g_\kappa)=1$, and again denote another independent generator by $\bar{g}_\kappa$. With this information, we can show that there is no nontrivial center. We examine a maximal commuting sub-algebra $\mathfrak{b}<\mfa$ generated by representatives $\mathfrak{b}=\langle  v_L, v_R,\Gamma_L(g_\kappa),\Gamma_R(g_\kappa)\rangle$. The result is a minimal representation of size $4$. See \tabref{tab:SPTneq1}.

\begin{table}
\begin{tabular}{|c|c|c|c|c|c|c|}
\hline 
 & $v_{L}v_{R}$ & $v_{L}$ & $v_{R}$ & $\Gamma_{L}\Gamma_{R}(\bar{g}_{\kappa})$ & $\Gamma_{L}(\bar{g}_{\kappa})$ & $\Gamma_{R}(\bar{g}_{\kappa})$\tabularnewline
\hline 
$\mid\psi\rangle $ & $u$ & $u'$ & $uu'$ & $w$ & $w'$ & $ww'$\tabularnewline
\hline 
$\Gamma_{L}(\bar{g}_{\kappa})$$\mid\psi\rangle $ & $-u$ & $-u'$ & $uu'$ & $-w$ & $-w'$ & $ww'$\tabularnewline
\hline 
$\Gamma_{R}(\bar{g}_{\kappa})$$\mid\psi\rangle $ & $-u$ & $u'$ & $-uu'$ & $-w$ & $w'$ & $-ww'$\tabularnewline
\hline 
$\Gamma_{L}\Gamma_{R}(\bar{g}_{\kappa})\mid\psi\rangle $ & $u$ & $-u'$ & $-uu'$ & $w$ & $-w'$ & $-ww'$\tabularnewline
\hline 
\end{tabular}
\caption{}\label{tab:SPTneq1}
\end{table}

\section{Twisted 1d representations}\label{app:twisted1dreps}
Given a symmetry group of the form $G= G'\times \mathbb{Z}^{\mathscr{T}}_2$ where $\mathbb{Z}^{\mathscr{T}}_2=\{1,T\}$ is the time reversal symmetry group with $\mathscr{T}^2=1$, and $G'$ is some unitary symmetry group, we wish to find all $\kappa(g)\in \text{U}(1)$ obeying 
\be\label{eq:twisted1dapp}
\kappa(gh)= \kappa(g)^{\alpha(h)}  \kappa(h)^{\alpha(g)}. 
\ee
with $\kappa(1)=1$. In particular, note that $1= \kappa(T^2) = \kappa(\mathscr{T})^{-2}$ so that $\eta \equiv\kappa(\mathscr{T})= \pm 1$. Note that for any $g',h' \in G'$
$$
\kappa(g'h')= \kappa(g')  \kappa(h') 
$$
Hence, restricted to $G'$, $\kappa$ is just some 1d representation $\chi$ of $G'$. Hence, each solution $\kappa$ determines an element $\chi \in H^1(G', \text{U}(1))$ and an $\eta \in \mathbb{Z}_2$. These two data determine $\kappa$ entirely: For general $g = (g', T^{ \sigma })$ with $\sigma=0, 1$ the defining relation \eqnref{eq:twisted1dapp} gives 
\be \label{eq:kappadef}
\kappa(g)= \chi(g')^{(-1)^{\sigma}} \eta^{\sigma} 
\ee
Let us now ensure that for any choice of $\chi,\eta$  there is a corresponding $\kappa$ solving \eqnref{eq:twisted1dapp}. Define $\kappa$ through \eqnref{eq:kappadef}.   For any $g=(g',T^{ \sigma }),h=(h',T^{ \tau})$ we have
\begin{align}
\kappa(gh) &=  \chi(g'h')^{(-1)^{\sigma+\tau} } \eta^{\sigma+\tau} \nonumber\\
&= \chi(g')^{(-1)^{\sigma+\tau}}\eta^{\sigma}  \chi(h')^{(-1)^{\sigma+\tau}}  \eta^{\tau} \nonumber\\
&= \chi(g')^{(-1)^{\sigma}\alpha(h)  }\eta^{\sigma \alpha(h)}  \chi(h')^{(-1)^{\tau} \alpha(g)}  \eta^{\tau  \alpha(g)} \nonumber\\
&= \kappa(g)^{\alpha(h)}  \kappa(h)^{\alpha(g)}\nonumber\\ \nonumber
\end{align}
Hence, the distinct solutions to  \eqnref{eq:twisted1dapp} correspond bijectively with the 1d representations of $G'$, and a certain $\mathbb{Z}_2$ index ($\eta=\pm1$. Hence, we say there is a $\mathcal{A}_G=H^{1}(G', U(1))\times \mathbb{Z}_2$ classification.
\end{appendix}
\input{SPTPRBFinal.bbl}
\end{document}

%% file: SPTPRBFinal.bbl
%

%% file: SPTPRBFinal.bbl
\begin{thebibliography}{51}%
\makeatletter
\providecommand \@ifxundefined [1]{%
 \@ifx{#1\undefined}
}%
\providecommand \@ifnum [1]{%
 \ifnum #1\expandafter \@firstoftwo
 \else \expandafter \@secondoftwo
 \fi
}%
\providecommand \@ifx [1]{%
 \ifx #1\expandafter \@firstoftwo
 \else \expandafter \@secondoftwo
 \fi
}%
\providecommand \natexlab [1]{#1}%
\providecommand \enquote  [1]{``#1''}%
\providecommand \bibnamefont  [1]{#1}%
\providecommand \bibfnamefont [1]{#1}%
\providecommand \citenamefont [1]{#1}%
\providecommand \href@noop [0]{\@secondoftwo}%
\providecommand \href [0]{\begingroup \@sanitize@url \@href}%
\providecommand \@href[1]{\@@startlink{#1}\@@href}%
\providecommand \@@href[1]{\endgroup#1\@@endlink}%
\providecommand \@sanitize@url [0]{\catcode `\\12\catcode `\$12\catcode
  `\&12\catcode `\#12\catcode `\^12\catcode `\_12\catcode `\%12\relax}%
\providecommand \@@startlink[1]{}%
\providecommand \@@endlink[0]{}%
\providecommand \url  [0]{\begingroup\@sanitize@url \@url }%
\providecommand \@url [1]{\endgroup\@href {#1}{\urlprefix }}%
\providecommand \urlprefix  [0]{URL }%
\providecommand \Eprint [0]{\href }%
\providecommand \doibase [0]{http://dx.doi.org/}%
\providecommand \selectlanguage [0]{\@gobble}%
\providecommand \bibinfo  [0]{\@secondoftwo}%
\providecommand \bibfield  [0]{\@secondoftwo}%
\providecommand \translation [1]{[#1]}%
\providecommand \BibitemOpen [0]{}%
\providecommand \bibitemStop [0]{}%
\providecommand \bibitemNoStop [0]{.\EOS\space}%
\providecommand \EOS [0]{\spacefactor3000\relax}%
\providecommand \BibitemShut  [1]{\csname bibitem#1\endcsname}%
\let\auto@bib@innerbib\@empty
\bibitem [{\citenamefont {Basko}\ \emph {et~al.}(2006)\citenamefont {Basko},
  \citenamefont {Aleiner},\ and\ \citenamefont {Altshuler}}]{Basko06}%
  \BibitemOpen
  \bibfield  {author} {\bibinfo {author} {\bibfnamefont {D.~M.}\ \bibnamefont
  {Basko}}, \bibinfo {author} {\bibfnamefont {I.~L.}\ \bibnamefont {Aleiner}},
  \ and\ \bibinfo {author} {\bibfnamefont {B.~L.}\ \bibnamefont {Altshuler}},\
  }\href {\doibase 10.1016/j.aop.2005.11.014} {\bibfield  {journal} {\bibinfo
  {journal} {Annals of Physics}\ }\textbf {\bibinfo {volume} {321}},\ \bibinfo
  {pages} {1126} (\bibinfo {year} {2006})}\BibitemShut {NoStop}%
\bibitem [{\citenamefont {Imbrie}(2014)}]{Imbrie14}%
  \BibitemOpen
  \bibfield  {author} {\bibinfo {author} {\bibfnamefont {J.~Z.}\ \bibnamefont
  {Imbrie}},\ }\href {http://arxiv.org/abs/1403.7837} {\bibfield  {journal}
  {\bibinfo  {journal} {arXiv:1403.7837 [cond-mat, physics:math-ph]}\ }
  (\bibinfo {year} {2014})},\ \bibinfo {note} {arXiv: 1403.7837}\BibitemShut
  {NoStop}%
\bibitem [{\citenamefont {Khemani}\ \emph {et~al.}(2015)\citenamefont
  {Khemani}, \citenamefont {Lazarides}, \citenamefont {Moessner},\ and\
  \citenamefont {Sondhi}}]{Khemani15}%
  \BibitemOpen
  \bibfield  {author} {\bibinfo {author} {\bibfnamefont {V.}~\bibnamefont
  {Khemani}}, \bibinfo {author} {\bibfnamefont {A.}~\bibnamefont {Lazarides}},
  \bibinfo {author} {\bibfnamefont {R.}~\bibnamefont {Moessner}}, \ and\
  \bibinfo {author} {\bibfnamefont {S.~L.}\ \bibnamefont {Sondhi}},\
  }\href@noop {} {\bibfield  {journal} {\bibinfo  {journal} {arXiv preprint
  arXiv:1508.03344}\ } (\bibinfo {year} {2015})}\BibitemShut {NoStop}%
\bibitem [{\citenamefont {Lazarides}\ \emph {et~al.}(2014)\citenamefont
  {Lazarides}, \citenamefont {Das},\ and\ \citenamefont
  {Moessner}}]{Lazarides14PRE}%
  \BibitemOpen
  \bibfield  {author} {\bibinfo {author} {\bibfnamefont {A.}~\bibnamefont
  {Lazarides}}, \bibinfo {author} {\bibfnamefont {A.}~\bibnamefont {Das}}, \
  and\ \bibinfo {author} {\bibfnamefont {R.}~\bibnamefont {Moessner}},\ }\href
  {\doibase 10.1103/PhysRevE.90.012110} {\bibfield  {journal} {\bibinfo
  {journal} {Phys. Rev. E}\ }\textbf {\bibinfo {volume} {90}},\ \bibinfo
  {pages} {012110} (\bibinfo {year} {2014})}\BibitemShut {NoStop}%
\bibitem [{\citenamefont {D'Alessio}\ and\ \citenamefont
  {Rigol}(2014)}]{Rigol14}%
  \BibitemOpen
  \bibfield  {author} {\bibinfo {author} {\bibfnamefont {L.}~\bibnamefont
  {D'Alessio}}\ and\ \bibinfo {author} {\bibfnamefont {M.}~\bibnamefont
  {Rigol}},\ }\href {\doibase 10.1103/PhysRevX.4.041048} {\bibfield  {journal}
  {\bibinfo  {journal} {Phys. Rev. X}\ }\textbf {\bibinfo {volume} {4}},\
  \bibinfo {pages} {041048} (\bibinfo {year} {2014})}\BibitemShut {NoStop}%
\bibitem [{\citenamefont {Abanin}\ \emph {et~al.}(2014)\citenamefont {Abanin},
  \citenamefont {De~Roeck},\ and\ \citenamefont {Huveneers}}]{Abanin14}%
  \BibitemOpen
  \bibfield  {author} {\bibinfo {author} {\bibfnamefont {D.}~\bibnamefont
  {Abanin}}, \bibinfo {author} {\bibfnamefont {W.}~\bibnamefont {De~Roeck}}, \
  and\ \bibinfo {author} {\bibfnamefont {F.}~\bibnamefont {Huveneers}},\
  }\href@noop {} {\bibfield  {journal} {\bibinfo  {journal} {arXiv preprint
  arXiv:1412.4752}\ } (\bibinfo {year} {2014})}\BibitemShut {NoStop}%
\bibitem [{\citenamefont {Ponte}\ \emph {et~al.}(2015)\citenamefont {Ponte},
  \citenamefont {Papi\ifmmode~\acute{c}\else \'{c}\fi{}}, \citenamefont
  {Huveneers},\ and\ \citenamefont {Abanin}}]{Ponte15}%
  \BibitemOpen
  \bibfield  {author} {\bibinfo {author} {\bibfnamefont {P.}~\bibnamefont
  {Ponte}}, \bibinfo {author} {\bibfnamefont {Z.}~\bibnamefont
  {Papi\ifmmode~\acute{c}\else \'{c}\fi{}}}, \bibinfo {author} {\bibfnamefont
  {F.}~\bibnamefont {Huveneers}}, \ and\ \bibinfo {author} {\bibfnamefont
  {D.~A.}\ \bibnamefont {Abanin}},\ }\href {\doibase
  10.1103/PhysRevLett.114.140401} {\bibfield  {journal} {\bibinfo  {journal}
  {Phys. Rev. Lett.}\ }\textbf {\bibinfo {volume} {114}},\ \bibinfo {pages}
  {140401} (\bibinfo {year} {2015})}\BibitemShut {NoStop}%
\bibitem [{\citenamefont {Lazarides}\ \emph {et~al.}(2015)\citenamefont
  {Lazarides}, \citenamefont {Das},\ and\ \citenamefont
  {Moessner}}]{Lazarides14PRL}%
  \BibitemOpen
  \bibfield  {author} {\bibinfo {author} {\bibfnamefont {A.}~\bibnamefont
  {Lazarides}}, \bibinfo {author} {\bibfnamefont {A.}~\bibnamefont {Das}}, \
  and\ \bibinfo {author} {\bibfnamefont {R.}~\bibnamefont {Moessner}},\ }\href
  {\doibase 10.1103/PhysRevLett.115.030402} {\bibfield  {journal} {\bibinfo
  {journal} {Phys. Rev. Lett.}\ }\textbf {\bibinfo {volume} {115}},\ \bibinfo
  {pages} {030402} (\bibinfo {year} {2015})}\BibitemShut {NoStop}%
\bibitem [{\citenamefont {Chen}\ \emph {et~al.}(2011)\citenamefont {Chen},
  \citenamefont {Gu},\ and\ \citenamefont {Wen}}]{Chen11}%
  \BibitemOpen
  \bibfield  {author} {\bibinfo {author} {\bibfnamefont {X.}~\bibnamefont
  {Chen}}, \bibinfo {author} {\bibfnamefont {Z.-C.}\ \bibnamefont {Gu}}, \ and\
  \bibinfo {author} {\bibfnamefont {X.-G.}\ \bibnamefont {Wen}},\ }\href
  {\doibase 10.1103/PhysRevB.84.235128} {\bibfield  {journal} {\bibinfo
  {journal} {Phys. Rev. B}\ }\textbf {\bibinfo {volume} {84}},\ \bibinfo
  {pages} {235128} (\bibinfo {year} {2011})}\BibitemShut {NoStop}%
\bibitem [{\citenamefont {Turner}\ \emph {et~al.}(2011)\citenamefont {Turner},
  \citenamefont {Pollmann},\ and\ \citenamefont {Berg}}]{PollmannBerg11}%
  \BibitemOpen
  \bibfield  {author} {\bibinfo {author} {\bibfnamefont {A.~M.}\ \bibnamefont
  {Turner}}, \bibinfo {author} {\bibfnamefont {F.}~\bibnamefont {Pollmann}}, \
  and\ \bibinfo {author} {\bibfnamefont {E.}~\bibnamefont {Berg}},\ }\href
  {\doibase 10.1103/PhysRevB.83.075102} {\bibfield  {journal} {\bibinfo
  {journal} {Phys. Rev. B}\ }\textbf {\bibinfo {volume} {83}},\ \bibinfo
  {pages} {075102} (\bibinfo {year} {2011})}\BibitemShut {NoStop}%
\bibitem [{\citenamefont {Chandran}\ \emph {et~al.}(2014)\citenamefont
  {Chandran}, \citenamefont {Khemani}, \citenamefont {Laumann},\ and\
  \citenamefont {Sondhi}}]{Chandran14}%
  \BibitemOpen
  \bibfield  {author} {\bibinfo {author} {\bibfnamefont {A.}~\bibnamefont
  {Chandran}}, \bibinfo {author} {\bibfnamefont {V.}~\bibnamefont {Khemani}},
  \bibinfo {author} {\bibfnamefont {C.~R.}\ \bibnamefont {Laumann}}, \ and\
  \bibinfo {author} {\bibfnamefont {S.~L.}\ \bibnamefont {Sondhi}},\ }\href
  {\doibase 10.1103/PhysRevB.89.144201} {\bibfield  {journal} {\bibinfo
  {journal} {Phys. Rev. B}\ }\textbf {\bibinfo {volume} {89}},\ \bibinfo
  {pages} {144201} (\bibinfo {year} {2014})}\BibitemShut {NoStop}%
\bibitem [{\citenamefont {Bahri}\ \emph {et~al.}(2015)\citenamefont {Bahri},
  \citenamefont {Vosk}, \citenamefont {Altman},\ and\ \citenamefont
  {Vishwanath}}]{Bahri15}%
  \BibitemOpen
  \bibfield  {author} {\bibinfo {author} {\bibfnamefont {Y.}~\bibnamefont
  {Bahri}}, \bibinfo {author} {\bibfnamefont {R.}~\bibnamefont {Vosk}},
  \bibinfo {author} {\bibfnamefont {E.}~\bibnamefont {Altman}}, \ and\ \bibinfo
  {author} {\bibfnamefont {A.}~\bibnamefont {Vishwanath}},\ }\href
  {http://dx.doi.org/10.1038/ncomms8341} {\bibfield  {journal} {\bibinfo
  {journal} {Nat Commun}\ }\textbf {\bibinfo {volume} {6}} (\bibinfo {year}
  {2015})}\BibitemShut {NoStop}%
\bibitem [{\citenamefont {Potter}\ and\ \citenamefont
  {Vishwanath}(2015)}]{Potter15}%
  \BibitemOpen
  \bibfield  {author} {\bibinfo {author} {\bibfnamefont {A.~C.}\ \bibnamefont
  {Potter}}\ and\ \bibinfo {author} {\bibfnamefont {A.}~\bibnamefont
  {Vishwanath}},\ }\href@noop {} {\bibfield  {journal} {\bibinfo  {journal}
  {arXiv preprint arXiv:1506.00592}\ } (\bibinfo {year} {2015})}\BibitemShut
  {NoStop}%
\bibitem [{\citenamefont {Serbyn}\ \emph
  {et~al.}(2013{\natexlab{a}})\citenamefont {Serbyn}, \citenamefont
  {Papi\ifmmode~\acute{c}\else \'{c}\fi{}},\ and\ \citenamefont
  {Abanin}}]{Serbyn13a}%
  \BibitemOpen
  \bibfield  {author} {\bibinfo {author} {\bibfnamefont {M.}~\bibnamefont
  {Serbyn}}, \bibinfo {author} {\bibfnamefont {Z.}~\bibnamefont
  {Papi\ifmmode~\acute{c}\else \'{c}\fi{}}}, \ and\ \bibinfo {author}
  {\bibfnamefont {D.~A.}\ \bibnamefont {Abanin}},\ }\href {\doibase
  10.1103/PhysRevLett.110.260601} {\bibfield  {journal} {\bibinfo  {journal}
  {Phys. Rev. Lett.}\ }\textbf {\bibinfo {volume} {110}},\ \bibinfo {pages}
  {260601} (\bibinfo {year} {2013}{\natexlab{a}})}\BibitemShut {NoStop}%
\bibitem [{\citenamefont {Serbyn}\ \emph
  {et~al.}(2013{\natexlab{b}})\citenamefont {Serbyn}, \citenamefont
  {Papi\ifmmode~\acute{c}\else \'{c}\fi{}},\ and\ \citenamefont
  {Abanin}}]{Serbyn13cons}%
  \BibitemOpen
  \bibfield  {author} {\bibinfo {author} {\bibfnamefont {M.}~\bibnamefont
  {Serbyn}}, \bibinfo {author} {\bibfnamefont {Z.}~\bibnamefont
  {Papi\ifmmode~\acute{c}\else \'{c}\fi{}}}, \ and\ \bibinfo {author}
  {\bibfnamefont {D.~A.}\ \bibnamefont {Abanin}},\ }\href {\doibase
  10.1103/PhysRevLett.111.127201} {\bibfield  {journal} {\bibinfo  {journal}
  {Phys. Rev. Lett.}\ }\textbf {\bibinfo {volume} {111}},\ \bibinfo {pages}
  {127201} (\bibinfo {year} {2013}{\natexlab{b}})}\BibitemShut {NoStop}%
\bibitem [{\citenamefont {Huse}\ \emph {et~al.}(2014)\citenamefont {Huse},
  \citenamefont {Nandkishore},\ and\ \citenamefont {Oganesyan}}]{Huse14}%
  \BibitemOpen
  \bibfield  {author} {\bibinfo {author} {\bibfnamefont {D.~A.}\ \bibnamefont
  {Huse}}, \bibinfo {author} {\bibfnamefont {R.}~\bibnamefont {Nandkishore}}, \
  and\ \bibinfo {author} {\bibfnamefont {V.}~\bibnamefont {Oganesyan}},\ }\href
  {\doibase 10.1103/PhysRevB.90.174202} {\bibfield  {journal} {\bibinfo
  {journal} {Phys. Rev. B}\ }\textbf {\bibinfo {volume} {90}},\ \bibinfo
  {pages} {174202} (\bibinfo {year} {2014})}\BibitemShut {NoStop}%
\bibitem [{Note1()}]{Note1}%
  \BibitemOpen
  \bibinfo {note} {A commuting stabilizer Hamiltonian takes form $H=\DOTSB
  \sum@ \slimits@ _r H_r$, where the $H_r$ are local and commute amongst
  themselves. See Ref.\protect \rev@citealpnum {Bahri15} for a well explained
  example of a commuting stabilizer SPT Hamiltonian, and examples in the main
  text.}\BibitemShut {Stop}%
\bibitem [{\citenamefont {Kitagawa}\ \emph {et~al.}(2010)\citenamefont
  {Kitagawa}, \citenamefont {Berg}, \citenamefont {Rudner},\ and\ \citenamefont
  {Demler}}]{Kitagawa10}%
  \BibitemOpen
  \bibfield  {author} {\bibinfo {author} {\bibfnamefont {T.}~\bibnamefont
  {Kitagawa}}, \bibinfo {author} {\bibfnamefont {E.}~\bibnamefont {Berg}},
  \bibinfo {author} {\bibfnamefont {M.}~\bibnamefont {Rudner}}, \ and\ \bibinfo
  {author} {\bibfnamefont {E.}~\bibnamefont {Demler}},\ }\href {\doibase
  10.1103/PhysRevB.82.235114} {\bibfield  {journal} {\bibinfo  {journal} {Phys.
  Rev. B}\ }\textbf {\bibinfo {volume} {82}},\ \bibinfo {pages} {235114}
  (\bibinfo {year} {2010})}\BibitemShut {NoStop}%
\bibitem [{\citenamefont {Jiang}\ \emph {et~al.}(2011)\citenamefont {Jiang},
  \citenamefont {Kitagawa}, \citenamefont {Alicea}, \citenamefont {Akhmerov},
  \citenamefont {Pekker}, \citenamefont {Refael}, \citenamefont {Cirac},
  \citenamefont {Demler}, \citenamefont {Lukin},\ and\ \citenamefont
  {Zoller}}]{Jiang11}%
  \BibitemOpen
  \bibfield  {author} {\bibinfo {author} {\bibfnamefont {L.}~\bibnamefont
  {Jiang}}, \bibinfo {author} {\bibfnamefont {T.}~\bibnamefont {Kitagawa}},
  \bibinfo {author} {\bibfnamefont {J.}~\bibnamefont {Alicea}}, \bibinfo
  {author} {\bibfnamefont {A.~R.}\ \bibnamefont {Akhmerov}}, \bibinfo {author}
  {\bibfnamefont {D.}~\bibnamefont {Pekker}}, \bibinfo {author} {\bibfnamefont
  {G.}~\bibnamefont {Refael}}, \bibinfo {author} {\bibfnamefont {J.~I.}\
  \bibnamefont {Cirac}}, \bibinfo {author} {\bibfnamefont {E.}~\bibnamefont
  {Demler}}, \bibinfo {author} {\bibfnamefont {M.~D.}\ \bibnamefont {Lukin}}, \
  and\ \bibinfo {author} {\bibfnamefont {P.}~\bibnamefont {Zoller}},\ }\href
  {\doibase 10.1103/PhysRevLett.106.220402} {\bibfield  {journal} {\bibinfo
  {journal} {Phys. Rev. Lett.}\ }\textbf {\bibinfo {volume} {106}},\ \bibinfo
  {pages} {220402} (\bibinfo {year} {2011})}\BibitemShut {NoStop}%
\bibitem [{\citenamefont {Lindner}\ \emph {et~al.}(2011)\citenamefont
  {Lindner}, \citenamefont {Refael},\ and\ \citenamefont
  {Galitski}}]{Lindner11}%
  \BibitemOpen
  \bibfield  {author} {\bibinfo {author} {\bibfnamefont {N.~H.}\ \bibnamefont
  {Lindner}}, \bibinfo {author} {\bibfnamefont {G.}~\bibnamefont {Refael}}, \
  and\ \bibinfo {author} {\bibfnamefont {V.}~\bibnamefont {Galitski}},\ }\href
  {\doibase 10.1038/nphys1926} {\bibfield  {journal} {\bibinfo  {journal}
  {Nature Physics}\ }\textbf {\bibinfo {volume} {7}},\ \bibinfo {pages} {490}
  (\bibinfo {year} {2011})}\BibitemShut {NoStop}%
\bibitem [{\citenamefont {Thakurathi}\ \emph {et~al.}(2013)\citenamefont
  {Thakurathi}, \citenamefont {Patel}, \citenamefont {Sen},\ and\ \citenamefont
  {Dutta}}]{Thakurathi13}%
  \BibitemOpen
  \bibfield  {author} {\bibinfo {author} {\bibfnamefont {M.}~\bibnamefont
  {Thakurathi}}, \bibinfo {author} {\bibfnamefont {A.~A.}\ \bibnamefont
  {Patel}}, \bibinfo {author} {\bibfnamefont {D.}~\bibnamefont {Sen}}, \ and\
  \bibinfo {author} {\bibfnamefont {A.}~\bibnamefont {Dutta}},\ }\href
  {\doibase 10.1103/PhysRevB.88.155133} {\bibfield  {journal} {\bibinfo
  {journal} {Phys. Rev. B}\ }\textbf {\bibinfo {volume} {88}},\ \bibinfo
  {pages} {155133} (\bibinfo {year} {2013})}\BibitemShut {NoStop}%
\bibitem [{\citenamefont {Rudner}\ \emph {et~al.}(2013)\citenamefont {Rudner},
  \citenamefont {Lindner}, \citenamefont {Berg},\ and\ \citenamefont
  {Levin}}]{Rudner13}%
  \BibitemOpen
  \bibfield  {author} {\bibinfo {author} {\bibfnamefont {M.~S.}\ \bibnamefont
  {Rudner}}, \bibinfo {author} {\bibfnamefont {N.~H.}\ \bibnamefont {Lindner}},
  \bibinfo {author} {\bibfnamefont {E.}~\bibnamefont {Berg}}, \ and\ \bibinfo
  {author} {\bibfnamefont {M.}~\bibnamefont {Levin}},\ }\href {\doibase
  10.1103/PhysRevX.3.031005} {\bibfield  {journal} {\bibinfo  {journal} {Phys.
  Rev. X}\ }\textbf {\bibinfo {volume} {3}},\ \bibinfo {pages} {031005}
  (\bibinfo {year} {2013})}\BibitemShut {NoStop}%
\bibitem [{\citenamefont {Asb\'oth}\ \emph {et~al.}(2014)\citenamefont
  {Asb\'oth}, \citenamefont {Tarasinski},\ and\ \citenamefont
  {Delplace}}]{Asboth14}%
  \BibitemOpen
  \bibfield  {author} {\bibinfo {author} {\bibfnamefont {J.~K.}\ \bibnamefont
  {Asb\'oth}}, \bibinfo {author} {\bibfnamefont {B.}~\bibnamefont
  {Tarasinski}}, \ and\ \bibinfo {author} {\bibfnamefont {P.}~\bibnamefont
  {Delplace}},\ }\href {\doibase 10.1103/PhysRevB.90.125143} {\bibfield
  {journal} {\bibinfo  {journal} {Phys. Rev. B}\ }\textbf {\bibinfo {volume}
  {90}},\ \bibinfo {pages} {125143} (\bibinfo {year} {2014})}\BibitemShut
  {NoStop}%
\bibitem [{\citenamefont {Carpentier}\ \emph {et~al.}(2015)\citenamefont
  {Carpentier}, \citenamefont {Delplace}, \citenamefont {Fruchart},\ and\
  \citenamefont {Gawedzki}}]{Carpentier15}%
  \BibitemOpen
  \bibfield  {author} {\bibinfo {author} {\bibfnamefont {D.}~\bibnamefont
  {Carpentier}}, \bibinfo {author} {\bibfnamefont {P.}~\bibnamefont
  {Delplace}}, \bibinfo {author} {\bibfnamefont {M.}~\bibnamefont {Fruchart}},
  \ and\ \bibinfo {author} {\bibfnamefont {K.}~\bibnamefont {Gawedzki}},\
  }\href {\doibase 10.1103/PhysRevLett.114.106806} {\bibfield  {journal}
  {\bibinfo  {journal} {Phys. Rev. Lett.}\ }\textbf {\bibinfo {volume} {114}},\
  \bibinfo {pages} {106806} (\bibinfo {year} {2015})}\BibitemShut {NoStop}%
\bibitem [{\citenamefont {Nathan}\ and\ \citenamefont
  {Rudner}(2015)}]{Nathan15}%
  \BibitemOpen
  \bibfield  {author} {\bibinfo {author} {\bibfnamefont {F.}~\bibnamefont
  {Nathan}}\ and\ \bibinfo {author} {\bibfnamefont {M.~S.}\ \bibnamefont
  {Rudner}},\ }\href {http://stacks.iop.org/1367-2630/17/i=12/a=125014}
  {\bibfield  {journal} {\bibinfo  {journal} {New Journal of Physics}\ }\textbf
  {\bibinfo {volume} {17}},\ \bibinfo {pages} {125014} (\bibinfo {year}
  {2015})}\BibitemShut {NoStop}%
\bibitem [{\citenamefont {{Roy}}\ and\ \citenamefont
  {{Harper}}(2016{\natexlab{a}})}]{Roy15}%
  \BibitemOpen
  \bibfield  {author} {\bibinfo {author} {\bibfnamefont {R.}~\bibnamefont
  {{Roy}}}\ and\ \bibinfo {author} {\bibfnamefont {F.}~\bibnamefont
  {{Harper}}},\ }\href@noop {} {\bibfield  {journal} {\bibinfo  {journal}
  {ArXiv e-prints}\ } (\bibinfo {year} {2016}{\natexlab{a}})},\ \Eprint
  {http://arxiv.org/abs/1603.06944} {arXiv:1603.06944 [cond-mat.str-el]}
  \BibitemShut {NoStop}%
\bibitem [{\citenamefont {Titum}\ \emph
  {et~al.}(2015{\natexlab{a}})\citenamefont {Titum}, \citenamefont {Lindner},
  \citenamefont {Rechtsman},\ and\ \citenamefont {Refael}}]{Titum15a}%
  \BibitemOpen
  \bibfield  {author} {\bibinfo {author} {\bibfnamefont {P.}~\bibnamefont
  {Titum}}, \bibinfo {author} {\bibfnamefont {N.~H.}\ \bibnamefont {Lindner}},
  \bibinfo {author} {\bibfnamefont {M.~C.}\ \bibnamefont {Rechtsman}}, \ and\
  \bibinfo {author} {\bibfnamefont {G.}~\bibnamefont {Refael}},\ }\href
  {\doibase 10.1103/PhysRevLett.114.056801} {\bibfield  {journal} {\bibinfo
  {journal} {Phys. Rev. Lett.}\ }\textbf {\bibinfo {volume} {114}},\ \bibinfo
  {pages} {056801} (\bibinfo {year} {2015}{\natexlab{a}})}\BibitemShut
  {NoStop}%
\bibitem [{\citenamefont {Titum}\ \emph
  {et~al.}(2015{\natexlab{b}})\citenamefont {Titum}, \citenamefont {Berg},
  \citenamefont {Rudner}, \citenamefont {Refael},\ and\ \citenamefont
  {Lindner}}]{Titum15b}%
  \BibitemOpen
  \bibfield  {author} {\bibinfo {author} {\bibfnamefont {P.}~\bibnamefont
  {Titum}}, \bibinfo {author} {\bibfnamefont {E.}~\bibnamefont {Berg}},
  \bibinfo {author} {\bibfnamefont {M.~S.}\ \bibnamefont {Rudner}}, \bibinfo
  {author} {\bibfnamefont {G.}~\bibnamefont {Refael}}, \ and\ \bibinfo {author}
  {\bibfnamefont {N.~H.}\ \bibnamefont {Lindner}},\ }\href@noop {} {\bibfield
  {journal} {\bibinfo  {journal} {arXiv preprint arXiv:1506.00650}\ } (\bibinfo
  {year} {2015}{\natexlab{b}})}\BibitemShut {NoStop}%
\bibitem [{\citenamefont {Schnyder}\ \emph {et~al.}(2008)\citenamefont
  {Schnyder}, \citenamefont {Ryu}, \citenamefont {Furusaki},\ and\
  \citenamefont {Ludwig}}]{Schnyder08}%
  \BibitemOpen
  \bibfield  {author} {\bibinfo {author} {\bibfnamefont {A.~P.}\ \bibnamefont
  {Schnyder}}, \bibinfo {author} {\bibfnamefont {S.}~\bibnamefont {Ryu}},
  \bibinfo {author} {\bibfnamefont {A.}~\bibnamefont {Furusaki}}, \ and\
  \bibinfo {author} {\bibfnamefont {A.~W.}\ \bibnamefont {Ludwig}},\
  }\href@noop {} {\bibfield  {journal} {\bibinfo  {journal} {Phys. Rev. B}\
  }\textbf {\bibinfo {volume} {78}},\ \bibinfo {pages} {195125} (\bibinfo
  {year} {2008})}\BibitemShut {NoStop}%
\bibitem [{\citenamefont {Fidkowski}\ and\ \citenamefont
  {Kitaev}(2011)}]{Kitaev11}%
  \BibitemOpen
  \bibfield  {author} {\bibinfo {author} {\bibfnamefont {L.}~\bibnamefont
  {Fidkowski}}\ and\ \bibinfo {author} {\bibfnamefont {A.}~\bibnamefont
  {Kitaev}},\ }\href {\doibase 10.1103/PhysRevB.83.075103} {\bibfield
  {journal} {\bibinfo  {journal} {Phys. Rev. B}\ }\textbf {\bibinfo {volume}
  {83}},\ \bibinfo {pages} {075103} (\bibinfo {year} {2011})}\BibitemShut
  {NoStop}%
\bibitem [{Note2()}]{Note2}%
  \BibitemOpen
  \bibinfo {note} {For fermionic states, the fermion parity of the symmetry
  action at the edge is also important\cite {Kitaev11}.}\BibitemShut {Stop}%
\bibitem [{Note3()}]{Note3}%
  \BibitemOpen
  \bibinfo {note} {See also Ref.~\protect \rev@citealpnum
  {Pollmann10,PollmannBerg11} for a more pedestrian exposition, and
  Ref.~\protect \rev@citealpnum {ChenScience} and Sec.~\ref {s:algebraic} for
  an introduction to cocycles.}\BibitemShut {Stop}%
\bibitem [{\citenamefont {{Else}}\ and\ \citenamefont
  {{Nayak}}(2016)}]{Else16}%
  \BibitemOpen
  \bibfield  {author} {\bibinfo {author} {\bibfnamefont {D.~V.}\ \bibnamefont
  {{Else}}}\ and\ \bibinfo {author} {\bibfnamefont {C.}~\bibnamefont
  {{Nayak}}},\ }\href@noop {} {\bibfield  {journal} {\bibinfo  {journal} {ArXiv
  e-prints}\ } (\bibinfo {year} {2016})},\ \Eprint
  {http://arxiv.org/abs/1602.04804} {arXiv:1602.04804 [cond-mat.str-el]}
  \BibitemShut {NoStop}%
\bibitem [{\citenamefont {{Potter}}\ \emph {et~al.}(2016)\citenamefont
  {{Potter}}, \citenamefont {{Morimoto}},\ and\ \citenamefont
  {{Vishwanath}}}]{Potter16}%
  \BibitemOpen
  \bibfield  {author} {\bibinfo {author} {\bibfnamefont {A.~C.}\ \bibnamefont
  {{Potter}}}, \bibinfo {author} {\bibfnamefont {T.}~\bibnamefont
  {{Morimoto}}}, \ and\ \bibinfo {author} {\bibfnamefont {A.}~\bibnamefont
  {{Vishwanath}}},\ }\href@noop {} {\bibfield  {journal} {\bibinfo  {journal}
  {ArXiv e-prints}\ } (\bibinfo {year} {2016})},\ \Eprint
  {http://arxiv.org/abs/1602.05194} {arXiv:1602.05194 [cond-mat.str-el]}
  \BibitemShut {NoStop}%
\bibitem [{\citenamefont {{Roy}}\ and\ \citenamefont
  {{Harper}}(2016{\natexlab{b}})}]{Harper16}%
  \BibitemOpen
  \bibfield  {author} {\bibinfo {author} {\bibfnamefont {R.}~\bibnamefont
  {{Roy}}}\ and\ \bibinfo {author} {\bibfnamefont {F.}~\bibnamefont
  {{Harper}}},\ }\href@noop {} {\bibfield  {journal} {\bibinfo  {journal}
  {ArXiv e-prints}\ } (\bibinfo {year} {2016}{\natexlab{b}})},\ \Eprint
  {http://arxiv.org/abs/1602.08089} {arXiv:1602.08089 [cond-mat.str-el]}
  \BibitemShut {NoStop}%
\bibitem [{\citenamefont {Else}\ \emph {et~al.}(2013)\citenamefont {Else},
  \citenamefont {Bartlett},\ and\ \citenamefont {Doherty}}]{Else13}%
  \BibitemOpen
  \bibfield  {author} {\bibinfo {author} {\bibfnamefont {D.~V.}\ \bibnamefont
  {Else}}, \bibinfo {author} {\bibfnamefont {S.~D.}\ \bibnamefont {Bartlett}},
  \ and\ \bibinfo {author} {\bibfnamefont {A.~C.}\ \bibnamefont {Doherty}},\
  }\href {\doibase 10.1103/PhysRevB.88.085114} {\bibfield  {journal} {\bibinfo
  {journal} {Phys. Rev. B}\ }\textbf {\bibinfo {volume} {88}},\ \bibinfo
  {pages} {085114} (\bibinfo {year} {2013})}\BibitemShut {NoStop}%
\bibitem [{\citenamefont {Gannot}(2015)}]{Gannot15}%
  \BibitemOpen
  \bibfield  {author} {\bibinfo {author} {\bibfnamefont {Y.}~\bibnamefont
  {Gannot}},\ }\href {http://arxiv.org/abs/1512.04190} {\bibfield  {journal}
  {\bibinfo  {journal} {arXiv:1512.04190 [cond-mat]}\ } (\bibinfo {year}
  {2015})},\ \bibinfo {note} {arXiv: 1512.04190}\BibitemShut {NoStop}%
\bibitem [{\citenamefont {Lieb}\ and\ \citenamefont {Robinson}()}]{Lieb72}%
  \BibitemOpen
  \bibfield  {author} {\bibinfo {author} {\bibfnamefont {E.~H.}\ \bibnamefont
  {Lieb}}\ and\ \bibinfo {author} {\bibfnamefont {D.~W.}\ \bibnamefont
  {Robinson}},\ }\href {\doibase 10.1007/BF01645779} {\bibfield  {journal}
  {\bibinfo  {journal} {Communications in Mathematical Physics}\ }\textbf
  {\bibinfo {volume} {28}},\ \bibinfo {pages} {251}}\BibitemShut {NoStop}%
\bibitem [{\citenamefont {Chen}\ and\ \citenamefont
  {Vishwanath}(2015)}]{Chen15}%
  \BibitemOpen
  \bibfield  {author} {\bibinfo {author} {\bibfnamefont {X.}~\bibnamefont
  {Chen}}\ and\ \bibinfo {author} {\bibfnamefont {A.}~\bibnamefont
  {Vishwanath}},\ }\href {\doibase 10.1103/PhysRevX.5.041034} {\bibfield
  {journal} {\bibinfo  {journal} {Phys. Rev. X}\ }\textbf {\bibinfo {volume}
  {5}},\ \bibinfo {pages} {041034} (\bibinfo {year} {2015})}\BibitemShut
  {NoStop}%
\bibitem [{Note4()}]{Note4}%
  \BibitemOpen
  \bibinfo {note} {As well as the fermion parity of the $\protect \mathaccentV
  {hat}05E{g}$ in fermionic systems.}\BibitemShut {Stop}%
\bibitem [{\citenamefont {Kitaev}(2003)}]{Kitaev03}%
  \BibitemOpen
  \bibfield  {author} {\bibinfo {author} {\bibfnamefont {A.~Y.}\ \bibnamefont
  {Kitaev}},\ }\href@noop {} {\bibfield  {journal} {\bibinfo  {journal} {Annals
  of Physics}\ }\textbf {\bibinfo {volume} {303}},\ \bibinfo {pages} {2}
  (\bibinfo {year} {2003})}\BibitemShut {NoStop}%
\bibitem [{Note5()}]{Note5}%
  \BibitemOpen
  \bibinfo {note} {Note that a clean variant of Eq.~\protect \textup {\hbox
  {\mathsurround \z@ \protect \normalfont (\ignorespaces \ref {eq:H1Z2T}\unskip
  \@@italiccorr )}} emerges in the high frequency expansion of the interacting
  drives considered in Ref.~\protect \rev@citealpnum {Iadecola15}.}\BibitemShut
  {Stop}%
\bibitem [{\citenamefont {Vasseur}\ \emph {et~al.}(2015)\citenamefont
  {Vasseur}, \citenamefont {Potter},\ and\ \citenamefont
  {Parameswaran}}]{Vasseur15}%
  \BibitemOpen
  \bibfield  {author} {\bibinfo {author} {\bibfnamefont {R.}~\bibnamefont
  {Vasseur}}, \bibinfo {author} {\bibfnamefont {A.~C.}\ \bibnamefont {Potter}},
  \ and\ \bibinfo {author} {\bibfnamefont {S.~A.}\ \bibnamefont
  {Parameswaran}},\ }\href {\doibase 10.1103/PhysRevLett.114.217201} {\bibfield
   {journal} {\bibinfo  {journal} {Phys. Rev. Lett.}\ }\textbf {\bibinfo
  {volume} {114}},\ \bibinfo {pages} {217201} (\bibinfo {year}
  {2015})}\BibitemShut {NoStop}%
\bibitem [{\citenamefont {{von Keyserlingk}}\ and\ \citenamefont
  {{Sondhi}}(2016)}]{vonKeyserlingkSondhi16b}%
  \BibitemOpen
  \bibfield  {author} {\bibinfo {author} {\bibfnamefont {C.~W.}\ \bibnamefont
  {{von Keyserlingk}}}\ and\ \bibinfo {author} {\bibfnamefont {S.~L.}\
  \bibnamefont {{Sondhi}}},\ }\href@noop {} {\bibfield  {journal} {\bibinfo
  {journal} {ArXiv e-prints}\ } (\bibinfo {year} {2016})},\ \Eprint
  {http://arxiv.org/abs/1602.06949} {arXiv:1602.06949 [cond-mat.str-el]}
  \BibitemShut {NoStop}%
\bibitem [{\citenamefont {Chen}\ \emph {et~al.}(2010)\citenamefont {Chen},
  \citenamefont {Gu},\ and\ \citenamefont {Wen}}]{Chen10}%
  \BibitemOpen
  \bibfield  {author} {\bibinfo {author} {\bibfnamefont {X.}~\bibnamefont
  {Chen}}, \bibinfo {author} {\bibfnamefont {Z.-C.}\ \bibnamefont {Gu}}, \ and\
  \bibinfo {author} {\bibfnamefont {X.-G.}\ \bibnamefont {Wen}},\ }\href
  {\doibase 10.1103/PhysRevB.82.155138} {\bibfield  {journal} {\bibinfo
  {journal} {Phys. Rev. B}\ }\textbf {\bibinfo {volume} {82}},\ \bibinfo
  {pages} {155138} (\bibinfo {year} {2010})}\BibitemShut {NoStop}%
\bibitem [{Note6()}]{Note6}%
  \BibitemOpen
  \bibinfo {note} {In the continuous time language, this is approximately the
  same as $\protect \mathcal {T} e^{-i \DOTSI \intop \ilimits@ ^T_0 dt
  H_{S(t)}}$ where $S(t)=[L+ct, R-ct]$, the notation $H_{S(t)}$ denotes those
  terms in the Hamiltonian involving only sites in region $S(t)$, and $c$ is
  the Lieb-Robinson velocity.}\BibitemShut {Stop}%
\bibitem [{Note7()}]{Note7}%
  \BibitemOpen
  \bibinfo {note} {In fermionic SPTs, fermion parity is always a symmetry, so
  as $W_{L,R}$ are symmetric they must also be fermion parity even\cite
  {Gu14}.}\BibitemShut {Stop}%
\bibitem [{\citenamefont {Pollmann}\ \emph {et~al.}(2010)\citenamefont
  {Pollmann}, \citenamefont {Turner}, \citenamefont {Berg},\ and\ \citenamefont
  {Oshikawa}}]{Pollmann10}%
  \BibitemOpen
  \bibfield  {author} {\bibinfo {author} {\bibfnamefont {F.}~\bibnamefont
  {Pollmann}}, \bibinfo {author} {\bibfnamefont {A.~M.}\ \bibnamefont
  {Turner}}, \bibinfo {author} {\bibfnamefont {E.}~\bibnamefont {Berg}}, \ and\
  \bibinfo {author} {\bibfnamefont {M.}~\bibnamefont {Oshikawa}},\ }\href
  {\doibase 10.1103/PhysRevB.81.064439} {\bibfield  {journal} {\bibinfo
  {journal} {Phys. Rev. B}\ }\textbf {\bibinfo {volume} {81}},\ \bibinfo
  {pages} {064439} (\bibinfo {year} {2010})}\BibitemShut {NoStop}%
\bibitem [{\citenamefont {Chen}\ \emph {et~al.}(2012)\citenamefont {Chen},
  \citenamefont {Gu}, \citenamefont {Liu},\ and\ \citenamefont
  {Wen}}]{ChenScience}%
  \BibitemOpen
  \bibfield  {author} {\bibinfo {author} {\bibfnamefont {X.}~\bibnamefont
  {Chen}}, \bibinfo {author} {\bibfnamefont {Z.-C.}\ \bibnamefont {Gu}},
  \bibinfo {author} {\bibfnamefont {Z.-X.}\ \bibnamefont {Liu}}, \ and\
  \bibinfo {author} {\bibfnamefont {X.-G.}\ \bibnamefont {Wen}},\ }\href
  {\doibase 10.1126/science.1227224} {\bibfield  {journal} {\bibinfo  {journal}
  {Science}\ }\textbf {\bibinfo {volume} {338}},\ \bibinfo {pages} {1604}
  (\bibinfo {year} {2012})},\ \Eprint
  {http://arxiv.org/abs/http://science.sciencemag.org/content/338/6114/1604.full.pdf}
  {http://science.sciencemag.org/content/338/6114/1604.full.pdf} \BibitemShut
  {NoStop}%
\bibitem [{\citenamefont {Iadecola}\ \emph {et~al.}(2015)\citenamefont
  {Iadecola}, \citenamefont {Santos},\ and\ \citenamefont
  {Chamon}}]{Iadecola15}%
  \BibitemOpen
  \bibfield  {author} {\bibinfo {author} {\bibfnamefont {T.}~\bibnamefont
  {Iadecola}}, \bibinfo {author} {\bibfnamefont {L.~H.}\ \bibnamefont
  {Santos}}, \ and\ \bibinfo {author} {\bibfnamefont {C.}~\bibnamefont
  {Chamon}},\ }\href {\doibase 10.1103/PhysRevB.92.125107} {\bibfield
  {journal} {\bibinfo  {journal} {Phys. Rev. B}\ }\textbf {\bibinfo {volume}
  {92}},\ \bibinfo {pages} {125107} (\bibinfo {year} {2015})}\BibitemShut
  {NoStop}%
\bibitem [{\citenamefont {Gu}\ and\ \citenamefont {Wen}(2014)}]{Gu14}%
  \BibitemOpen
  \bibfield  {author} {\bibinfo {author} {\bibfnamefont {Z.-C.}\ \bibnamefont
  {Gu}}\ and\ \bibinfo {author} {\bibfnamefont {X.-G.}\ \bibnamefont {Wen}},\
  }\href {\doibase 10.1103/PhysRevB.90.115141} {\bibfield  {journal} {\bibinfo
  {journal} {Phys. Rev. B}\ }\textbf {\bibinfo {volume} {90}},\ \bibinfo
  {pages} {115141} (\bibinfo {year} {2014})}\BibitemShut {NoStop}%
\end{thebibliography}
